\documentclass[reqno]{amsart}

\usepackage{amsmath}
\usepackage{amssymb}
\usepackage{amsthm}
\usepackage{bm}
\usepackage{color}
\usepackage{eucal}
\usepackage{fullpage}
\usepackage{graphicx}
\usepackage{hhline}
\usepackage{lineno}
\usepackage[sc]{mathpazo}
\usepackage{mathrsfs}
\usepackage{mathtools}
\usepackage[numbers,sort&compress,square]{natbib}
\usepackage{setspace}
\usepackage{subfig}
\usepackage{xr}
\usepackage[all,cmtip]{xy}

\title{Environmental fitness heterogeneity in the Moran process}
\author{Kamran Kaveh$^{\ast}$}
\author{Alex McAvoy$^{\ast}$}
\author{Martin A. Nowak}
\thanks{$^{\ast}$K. K. and A. M. contributed equally to this study. Corresponding author: Martin A. Nowak (\texttt{martin\_nowak@harvard.edu})}

\newcommand*\patchAmsMathEnvironmentForLineno[1]{%
\expandafter\let\csname old#1\expandafter\endcsname\csname #1\endcsname
\expandafter\let\csname oldend#1\expandafter\endcsname\csname end#1\endcsname
\renewenvironment{#1}%
{\linenomath\csname old#1\endcsname}%
{\csname oldend#1\endcsname\endlinenomath}}% 
\newcommand*\patchBothAmsMathEnvironmentsForLineno[1]{%
\patchAmsMathEnvironmentForLineno{#1}%
\patchAmsMathEnvironmentForLineno{#1*}}%
\AtBeginDocument{%
\patchBothAmsMathEnvironmentsForLineno{equation}%
\patchBothAmsMathEnvironmentsForLineno{align}%
\patchBothAmsMathEnvironmentsForLineno{flalign}%
\patchBothAmsMathEnvironmentsForLineno{alignat}%
\patchBothAmsMathEnvironmentsForLineno{gather}%
\patchBothAmsMathEnvironmentsForLineno{multline}%
}

\theoremstyle{definition}

\newtheorem{lemma}{Lemma}

\newtheorem{remark}{Remark}

\newcommand{\eq}[1]{\textbf{Eq.~\ref{eq:#1}}}
\newcommand{\fig}[1]{\textbf{Fig.~\ref{fig:#1}}}

\begin{document}

\allowdisplaybreaks

\begin{abstract}
Many mathematical models of evolution assume that all individuals experience the same environment. Here, we study the Moran process in heterogeneous environments. The population is of finite size with two competing types, which are exposed to a fixed number of environmental conditions. Reproductive rate is determined by both the type and the environment. We first calculate the condition for selection to favor the mutant relative to the resident wild type. In large populations, the mutant is favored if and only if the mutant's spatial average reproductive rate exceeds that of the resident. But environmental heterogeneity elucidates an interesting asymmetry between the mutant and the resident. Specifically, mutant heterogeneity suppresses its fixation probability; if this heterogeneity is strong enough, it can even completely offset the effects of selection (including in large populations). In contrast, resident heterogeneity has no effect on a mutant's fixation probability in large populations and can amplify it in small populations.
\end{abstract}

\maketitle

\section{Introduction}
Evolutionary dynamics deals with the appearance and competition of traits over time. The success of an initially-rare mutant arising in a population depends on a number of factors, including the population's spatial structure and the mutant's reproductive fitness relative to the resident. One quantitative measure of a mutant's success is its fixation probability, which describes the chance that the mutant's lineage will take over the population \citep{nagylaki:S:1992}. The effect of a particular property of the population (such as its spatial structure) on natural selection is often measured directly in terms of its effects on this probability of fixation. Among the many noted demographic features that affect evolutionary outcomes, comparatively little is known about the effects of environmental heterogeneity in reproductive fitness on evolutionary dynamics.

One source of interaction and migration heterogeneity is population structure. \citet{lieberman:Nature:2005} use graphs as a model for population structure and show that ``isothermal" structures do not alter fixation probabilities under birth-death updating, expanding upon a related observation for subdivided populations \citep{maruyama:TPB:1974}. Non-isothermal graphs can change this fixation probability and, in particular, act as amplifiers or suppressors of selection--a topic of considerable current interest \citep{lieberman:Nature:2005,ohtsuki:PRSB:2006,broom:PRSA:2008,patwa:JRSI:2008,broom:JIM:2009,houchmandzadeh:NJP:2011,mertzios:TCS:2013,monk:PRSA:2014,adlam:SR:2014,allen:PLOSCB:2015,adlam:PRSA:2015}. Recent work suggests that randomness in dispersal patterns yields either amplifiers or suppressors of selection \citep{antal:PRL:2006,sood:PRE:2008,manem:JTB:2014,adlam:SR:2014,hindersin:PLOSCB:2015}. Although spatial structure and frequency-dependent fitness have been incorporated into many evolutionary models, their effects on evolutionary dynamics are not fully understood. Even less is known about the effects of environmental heterogeneity, which can affect fitness through a non-uniform distribution of resources.

Despite the fact that there is still much left to be understood about the effects of environmental heterogeneity, its importance in theoretical models has long been recognized, particularly in population genetics \citep{gillespie:OUP:1991,barton:GR:1993,barton:G:1995}. More than sixty years ago, \citet{levene:AN:1953} introduced a diploid model in which two alleles are favored in different ecological niches and showed that genetic equilibrium is possible even when there is no niche in which the heterozygote is favored over both homozygotes. \citet{haldane:JG:1963} subsequently treated a temporal analogue of this fitness asymmetry, which was then incorporated into a study of polymorphism under both spatial and temporal fitness heterogeneity \citep{ewing:AN:1979}. \citet{arnold:AN:1983} described the spatial model of \citet{levene:AN:1953} as ``the beginning of theoretical ecological genetics."

Many studies of environmental heterogeneity have focused largely on metapopulation or island models under weak selective pressure, inspired by the evolution of habitat-specialist traits in heterogeneous environments \citep{levins:AN:1962,levins:AN:1963,pollak:JAP:1966,schreiber:AN:2009,vuilleumier:TPB:2010}. These metapopulation models assume connected islands (habitats) where migration is allowed between islands, and environmental heterogeneity is parametrized by a variable fitness difference between two competing types and assumed to be small (i.e. weak selection). Notably, in the limit of strong connectivity between islands, variations in fitness advantage do not affect fixation probability \citep{nagylaki:JMB:1980}. Others address the issue of fixation in two-island \citep{gavrilets:PE:2002} and multi-habitat \citep{whitlock:G:2005} models with variable fitness.

A more fine-grained heterogeneity requires an extension of the stepping-stone models to evolutionary graphs \citep{masuda:PRE:2010,hauser:JTB:2014,maciejewski:JTB:2014,maciejewski:PLoSCB:2014}. So far, much of the work in this area has been done through numerical simulations of specific structures and fitness distributions. For example, \citet{manem:PLOSONE:2015} demonstrated via death-birth simulations on a structured mesh that heterogeneity in the fitness distribution can decrease the fixation probability of a beneficial mutant. \citet{hauser:JTB:2014}, through exact calculations for small populations and simulations for larger populations, showed that heterogeneity in background fitness suppresses selection. Using an interesting analytical approach, \citet{masuda:PRE:2010} estimated the scaling behavior of the average consensus time in a voter model for random environments with uniform or power-law fitness distributions. More recently, \citet{mahdipour:SR:2017} considered a death-birth model on a cycle with random background fitness. Using numerical simulations, they observed that heterogeneity leads to an increase in fixation probability. However, in the same model, heterogeneity has also been shown to increase the time to fixation \citep{farhang:PLOSCB:2017}.

\citet{taylor:EJP:2007} distilled much of the research into heterogeneity with the remark that ``[o]ne of the key insights to emerge from population genetics theory is that the effectiveness of natural selection is reduced by random variation in individual survival and reproduction." However, beyond the fact that the Wright-Fisher model is the standard paradigm for many of these works in population genetics, results on environmental heterogeneity often rely on assumptions such as weak selection or restrictions on population structure or migration rates.

In this study, we take a different approach and consider the environmental heterogeneity in the Moran process with no restrictions on selection intensity. The Moran process models an idealized population of constant, finite size, $N$, with two competing types, $A$ and $B$ \citep{moran:MPCPS:1958}. At each time step, an individual is chosen for birth with probability proportional to reproductive fitness (which can depend on both the individual's type and the environment in which they reside), and the resulting offspring replaces a random individual in the population. One key difference between the Moran and Wright-Fisher models, which are both well-established in theoretical biology, is that generations overlap in the former but not in the latter. This aspect of the Moran model, which has been noted to result in qualitative differences in the dynamics \citep{bhaskar:B:2009,proulx:TPB:2011}, also has the added benefit of making some calculations (such as of a mutant's fixation probability) exact for the Moran process that are only approximations under Wright-Fisher updating \citep{novozhilov:BB:2006}.

We focus on the following questions for the Moran process:
\begin{itemize}
\item Can we predict the fate of a random mutant in a heterogeneous environment, given the measures of heterogeneity such as the standard deviation of mutant (and resident) fitness values?
\item Is the effect of environmental heterogeneity asymmetric with respect to the types? In other words, does variability in environmental conditions affect mutants more than residents?
\item What are the finite-population effects on fixation probability in a heterogeneous environment?
\item What is the interplay between dispersal structure and the environmental fitness distribution?
\end{itemize}

Through explicit formulas for fixation probabilities in large populations, we show that selection favors the mutant type if and only if the expected fitness of a randomly-placed mutant exceeds that of a randomly-placed resident. In other words, the mutant type is neutral relative to the resident if and only if the arithmetic mean of all possible fitness values for the mutant is the same as that of the resident. We also consider this selection condition in smaller populations, where we demonstrate how a combination of heterogeneity and drift results in a much more complicated criterion for the mutant to be favored over the resident.

More importantly, we show that mutant heterogeneity categorically suppresses selection; in particular, any such heterogeneity decreases the fixation probability of a beneficial mutant. In contrast, heterogeneity in resident fitness does not change a mutant's fixation probability when the population size is large, and it can even amplify selection in small populations. These observations uncover an asymmetry between the mutant and resident types in heterogeneous environments. Furthermore, since we impose no restrictions on selection intensity, our results highlight behavior that is difficult to see under weak heterogeneity.

\section{Model and fixation probabilities}\label{sec:model}
Consider a population of size $N$ in which each individual has one of two types, $A$ (mutant) or $B$ (resident). There are $m$ different environments in which an individual can reside, and we denote by $N_{i}$ the size of environment $i$ (meaning the number of individuals, of any type, that can reside in environment $i$) for $i=1,\dots ,m$. In environment $i$, $A$ has relative fitness $a_{i}$ and $B$ has relative fitness $b_{i}$. At each time step, an individual is chosen for reproduction with probability proportional to relative fitness. An individual subsequently dies (uniformly-at-random) and is replaced by the new offspring (see \textbf{Fig. \ref{fig:processDefinition}}).

\begin{figure}
\centering
\includegraphics[width=0.8\textwidth]{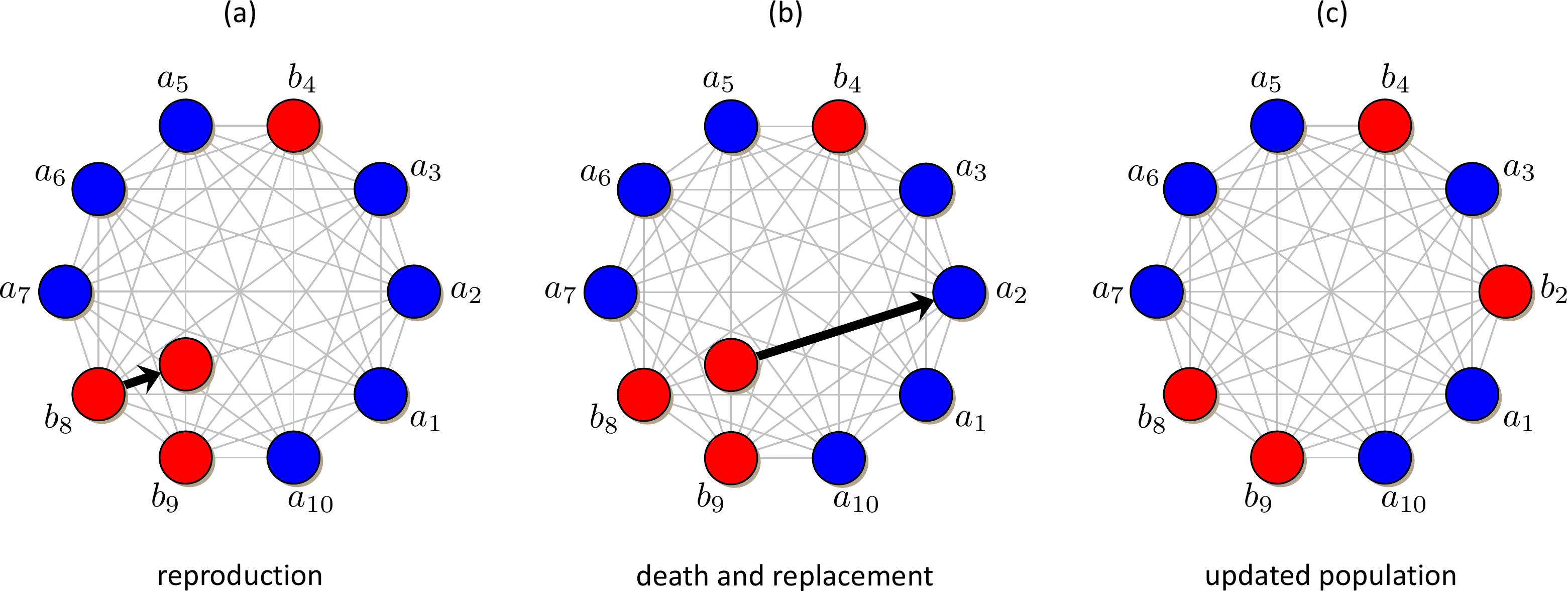}
\caption{Birth-death updating with environmental heterogeneity in reproductive fitness. At location $i$, an $A$-individual (mutant) has fitness $a_{i}$ and a $B$-individual (resident) has fitness $b_{i}$. At each time step, an individual is selected to reproduce with probability proportional to fitness; the offspring then replaces a random individual chosen for death. Here, the individual at location $8$ reproduces and its offspring replaces the individual at location $2$. Although the parent has fitness $b_{8}$, the offspring has fitness $b_{2}$ since it is in a different environment. While dispersal is determined by a complete graph (light grey), the population cannot be considered ``unstructured" since one must keep track of locations due to environmental variations in fitness (which could result from variations in resources).\label{fig:processDefinition}}
\end{figure}

The fraction of each fitness value present in the population defines mass functions, $f_{N}\left(a\right)$ and $g_{N}\left(b\right)$. That is, if there are $m$ environments with fitness values $a_{i}$ and $b_{i}$ in environment $i\in\left\{1,\dots ,m\right\}$, then
\begin{subequations}
\begin{align}
f_{N}\left(a\right) &= \begin{cases}\frac{N_{i}}{N} & a=a_{i}\textrm{ for some }i\in\left\{1,\dots ,m\right\} , \\ 0 & \textrm{otherwise} ;\end{cases} \\
g_{N}\left(b\right) &= \begin{cases}\frac{N_{i}}{N} & b=b_{i}\textrm{ for some }i\in\left\{1,\dots ,m\right\} , \\ 0 & \textrm{otherwise} .\end{cases}
\end{align}
\end{subequations}
We let $\overline{a}\coloneqq\frac{1}{N}\sum_{i=1}^{m}N_{i}a_{i}$ and $\overline{b}\coloneqq\frac{1}{N}\sum_{i=1}^{m}N_{i}b_{i}$ be the environmental fitness averages for $A$ and $B$, respectively; that is, $\overline{a}$ (resp. $\overline{b}$) is the expected fitness of a randomly-placed individual of type $A$ (resp. $B$). The classical Moran process \citep{moran:MPCPS:1958} is recovered when $a_{i}=\overline{a}$ and $b_{i}=\overline{b}$ for every $i$ (or, equivalently, when $m=1$ and $N_{1}=N$).

Without fitness heterogeneity, the state of the population is completely determined by the number of individuals of type $A$. Let $\rho_{A}^{N}$ be the probability that a single mutant ($A$), initialized uniformly-at-random in the population, fixates when the remaining $N-1$ individuals are of the resident type ($B$). Similarly, let $\rho_{B}^{N}$ be the probability that a single, randomly-placed resident ($B$) fixates in a population of $N-1$ mutants ($A$). A standard way to measure the evolutionary success of $A$ relative to $B$ is to compare $\rho_{A}^{N}$ to $\rho_{B}^{N}$. Type $A$ is favored over $B$ if $\rho_{A}^{N}>\rho_{B}^{N}$, disfavored relative to $B$ if $\rho_{A}^{N}<\rho_{B}^{N}$, and neutral relative to $B$ if $\rho_{A}^{N}=\rho_{B}^{N}$ \citep{tarnita:JTB:2009}. The equation $\rho_{A}^{N}=\rho_{B}^{N}$ is the ``neutrality condition" for fixation probability.

Suppose that $\overline{a}$ and $\overline{b}$ are the fitness values of $A$ and $B$, respectively, in the classical Moran process. Since there is no heterogeneity in the environment, one can think of $\rho_{A}^{N}=\rho_{A}^{N}\left(\overline{a},\overline{b}\right)$ and $\rho_{B}^{N}=\rho_{B}^{N}\left(\overline{a},\overline{b}\right)$ as functions of $\overline{a}$ and $\overline{b}$. Furthermore, $\rho_{B}^{N}\left(\overline{a},\overline{b}\right) =\rho_{A}^{N}\left(\overline{b},\overline{a}\right)$ since $A$ and $B$ are distinguished by only their fitness. Therefore, $A$ is neutral with respect to $B$ if and only if $\rho_{A}^{N}\left(\overline{a},\overline{b}\right) =\rho_{A}^{N}\left(\overline{b},\overline{a}\right)$. Since we know
\begin{align}
\rho_{A}^{N}\left(\overline{a},\overline{b}\right) &= \begin{cases}\frac{1-\overline{b}/\overline{a}}{1-\left(\overline{b}/\overline{a}\right)^{N}} & \overline{a}\neq\overline{b} ; \\ 1/N & \overline{a}=\overline{b}\end{cases}
\end{align}
\citep[see][]{nowak:BP:2006}, one can see that $\rho_{A}^{N}\left(\overline{a},\overline{b}\right) =\rho_{A}^{N}\left(\overline{b},\overline{a}\right)$ if and only if $\overline{a}=\overline{b}$, which makes intuitive sense because then $A$ is neutral relative to $B$ if and only if the two types are indistinguishable from a fitness standpoint.

In the Moran process with fitness heterogeneity, the fixation probability of a single $A$-individual could depend on its environment, so it is important to account for this initial environment when considering an analogue of the neutrality condition. Let $\mathbf{e}_{i}$ denote the state in which all individuals have type $B$ except for one individual of type $A$ in environment $i$. Let $\mathbf{A}$ be the monomorphic state in which all individuals have type $A$. We denote by $\rho_{\mathbf{e}_{i},\mathbf{A}}^{N}\left(f_{N},g_{N}\right)$ the probability that, when starting from this rare-mutant state, the $A$ type eventually takes over the population. Let $\rho_{A}^{N}\left(f_{N},g_{N}\right)$ be the fixation probability of an $A$-individual, averaged over all $N$ initial locations of the mutant, i.e.
\begin{align}
\rho_{A}^{N}\left(f_{N},g_{N}\right) &\coloneqq \frac{1}{N}\sum_{i=1}^{m}N_{i}\rho_{\mathbf{e}_{i},\mathbf{A}}^{N}\left(f_{N},g_{N}\right) .
\end{align}
A natural extension of the comparison of $\rho_{A}^{N}\left(a,b\right)$ to $\rho_{A}^{N}\left(b,a\right)$ is the comparison of $\rho_{A}^{N}\left(f_{N},g_{N}\right)$ to $\rho_{A}^{N}\left(g_{N},f_{N}\right)$. In other words, the neutrality condition is then defined by the equation $\rho_{A}^{N}\left(f_{N},g_{N}\right) =\rho_{A}^{N}\left(g_{N},f_{N}\right)$. We now turn to analyzing this neutrality condition for two types of populations: \textit{(i)} small populations, where drift plays a significant role in the dynamics, and \textit{(ii)} the large-population limit, where selection dominates.

\subsection{Small populations}
When $N$ is small, we cannot ignore the effects of random drift and, consequently, we do not expect the neutrality condition to be as simple as it is when $N$ is large (where one can focus on the effects of selection only). When $N=2$, there is environmental heterogeneity if there are $m=2$ environments (otherwise the model is the classical Moran process). For such a small population, it is simple to directly solve the standard recurrence equations for fixation probabilities (\ref{sec:appendixA}) to get
\begin{subequations}
\begin{align}
\rho_{A}^{N}\left(f_{N},g_{N}\right) &= \frac{1}{2}\left(\frac{a_{1}}{a_{1}+b_{2}}+\frac{a_{2}}{a_{2}+b_{1}}\right) ; \\
\rho_{A}^{N}\left(g_{N},f_{N}\right) &= \frac{1}{2}\left(\frac{b_{1}}{b_{1}+a_{2}}+\frac{b_{2}}{b_{2}+a_{1}}\right) .
\end{align}
\end{subequations}
The neutrality condition in this case is equivalent to $a_{1}a_{2}=b_{1}b_{2}$ (i.e. $\sqrt{a_{1}a_{2}}=\sqrt{b_{1}b_{2}}$).

On the other hand, even $N=3$ demonstrates how the neutrality condition quickly gets complicated for small values of $N$ greater than $2$. Again, for $N=3$, we can solve directly for fixation probabilities, $\rho$, but their expressions are complicated and not especially easy to interpret. Under the simplifying assumption $b_{1}=b_{2}=b_{3}=1$, the neutrality condition is equivalent to
\begin{align}
0 &= 6a_{1}^{3}a_{2}^{2}a_{3}+4a_{1}^{3}a_{2}^{2}+6a_{1}^{3}a_{2}a_{3}^{2}+14a_{1}^{3}a_{2}a_{3}+5a_{1}^{3}a_{2}+4a_{1}^{3}a_{3}^{2}+5a_{1}^{3}a_{3} \nonumber \\
&\quad +6a_{1}^{2}a_{2}^{3}a_{3}+4a_{1}^{2}a_{2}^{3}+12a_{1}^{2}a_{2}^{2}a_{3}^{2}+34a_{1}^{2}a_{2}^{2}a_{3}+14a_{1}^{2}a_{2}^{2}+6a_{1}^{2}a_{2}a_{3}^{3} \nonumber \\
&\quad +34a_{1}^{2}a_{2}a_{3}^{2}+41a_{1}^{2}a_{2}a_{3}+4a_{1}^{2}a_{3}^{3}+14a_{1}^{2}a_{3}^{2}-16a_{1}^{2}+6a_{1}a_{2}^{3}a_{3}^{2}+14a_{1}a_{2}^{3}a_{3} \nonumber \\
&\quad +5a_{1}a_{2}^{3}+6a_{1}a_{2}^{2}a_{3}^{3}+34a_{1}a_{2}^{2}a_{3}^{2}+41a_{1}a_{2}^{2}a_{3}+14a_{1}a_{2}a_{3}^{3}+41a_{1}a_{2}a_{3}^{2} \nonumber \\
&\quad -49a_{1}a_{2}+5a_{1}a_{3}^{3}-49a_{1}a_{3}-56a_{1}+4a_{2}^{3}a_{3}^{2}+5a_{2}^{3}a_{3}+4a_{2}^{2}a_{3}^{3} \nonumber \\
&\quad +14a_{2}^{2}a_{3}^{2}-16a_{2}^{2}+5a_{2}a_{3}^{3}-49a_{2}a_{3}-56a_{2}-16a_{3}^{2}-56a_{3}-48 .
\end{align}

For larger (but still finite $N$), the neutrality condition grows only more complicated. Therefore, in the following section, we turn to analyzing this neutrality condition in the large-population limit.

\subsection{Large-population limit}
Suppose that $m$ is fixed and that the size of environment $i$ is a function of the overall population size, $N$, and that there exists $\left(p_{1},\dots ,p_{m}\right)\in\left(0,1\right)^{m}$ such that environment $i$ satisfies $\lim_{N\rightarrow\infty}\frac{N_{i}\left(N\right)}{N}=p_{i}$ for every $i=1,\dots ,m$. (Note that $N_{i}$ can be an arbitrary function of $N$ as long as it is positive, integer-valued, and satisfies $\lim_{N\rightarrow\infty}\frac{N_{i}\left(N\right)}{N}=p_{i}\in\left(0,1\right)$.) Under this assumption, the mass functions $f_{N}$ and $g_{N}$ have well-defined limits, $f\coloneqq\lim_{N\rightarrow\infty}f_{N}$ and $g\coloneqq\lim_{N\rightarrow\infty}g_{N}$, respectively. Let $\overline{a}=\sum_{i=1}^{m}p_{i}a_{i}$ and $\overline{b}=\sum_{i=1}^{m}p_{i}b_{i}$ be the mean fitness values of the mutant type and the resident type, respectively, with respect to these distributions.

Let $\mathbb{E}_{f}$ denote the expectation with respect to the mass function $f$. We show in \ref{sec:appendixA} that, when we take $N\rightarrow\infty$, the limiting value of the fixation probability of a randomly-placed mutant, $\rho_{A}^{\infty}\left(f,g\right) \coloneqq\lim_{N\rightarrow\infty}\rho_{A}^{N}\left(f_{N},g_{N}\right)$, satisfies the following equation:
\begin{subequations}\label{eq:integralEquation}
\begin{align}
\rho_{A}^{\infty}\left(f,g\right) &= 0 \quad\textrm{if}\quad \overline{a}\leqslant\overline{b} ; \label{eq:integralEquationA} \\
\mathbb{E}_{f}\left[ \frac{a}{\overline{b}+a\rho_{A}^{\infty}\left(f,g\right)} \right] &= 1 \quad\textrm{if}\quad \overline{a}>\overline{b} . \label{eq:integralEquationB}
\end{align}
\end{subequations}
Therefore, $\rho_{A}^{\infty}\left(f,g\right) =\rho_{A}^{\infty}\left(g,f\right)$ if and only if $\overline{a}=\overline{b}$, which gives the neutrality condition for large populations.

From the neutrality condition for large populations, we also obtain conditions for selection to favor or disfavor the mutant type: $A$ is favored relative to $B$ if and only if $\overline{a}>\overline{b}$, and $A$ is disfavored relative to $B$ if and only if $\overline{a}<\overline{b}$. Therefore, the performance of one type relative to another can be deduced from the classical (homogeneous) model by replacing each location's fitness values, $a_{i}$ and $b_{i}$, by the spatial averages, $\overline{a}$ and $\overline{b}$. Although one can make a rough comparison of two types by looking at their mean fitness values, we show in the next section that mutant heterogeneity acts further as a suppressor of selection.

\section{Heterogeneity in mutant fitness}
In this section, we look at what happens to an invading mutant's fixation probability if its heterogeneous fitness values are replaced by their spatial average. Note that there is no heterogeneity in mutant (resp. resident) fitness if $f\left(\overline{a}\right) =1$ (resp. $g\left(\overline{b}\right) =1$). If either of these conditions holds, then we replace $f$ by $\overline{a}$ (resp. $g$ by $\overline{b}$) in the notation $\rho_{A}^{\infty}\left(f,g\right)$. For example, $\rho_{A}^{\infty}\left(f,\overline{b}\right)$ denotes the limiting value of $A$'s fixation probability when \textit{(i)} the fitness of $A$ is distributed according to $f$ and \textit{(ii)} every resident type has fitness exactly $\overline{b}$ (i.e. there is no resident heterogeneity). The first thing to notice is that, from \eq{integralEquation}, we have $\rho_{A}^{\infty}\left(f,g\right) =\rho_{A}^{\infty}\left(f,\overline{b}\right)$, so environmental heterogeneity of the resident does not affect the fixation probability of the mutant in the large-population limit. We next turn to how $\rho_{A}^{\infty}\left(f,g\right)$ compares to $\rho_{A}^{\infty}\left(\overline{a},g\right)$:

\begin{figure}
\centering
\includegraphics[width=0.8\textwidth]{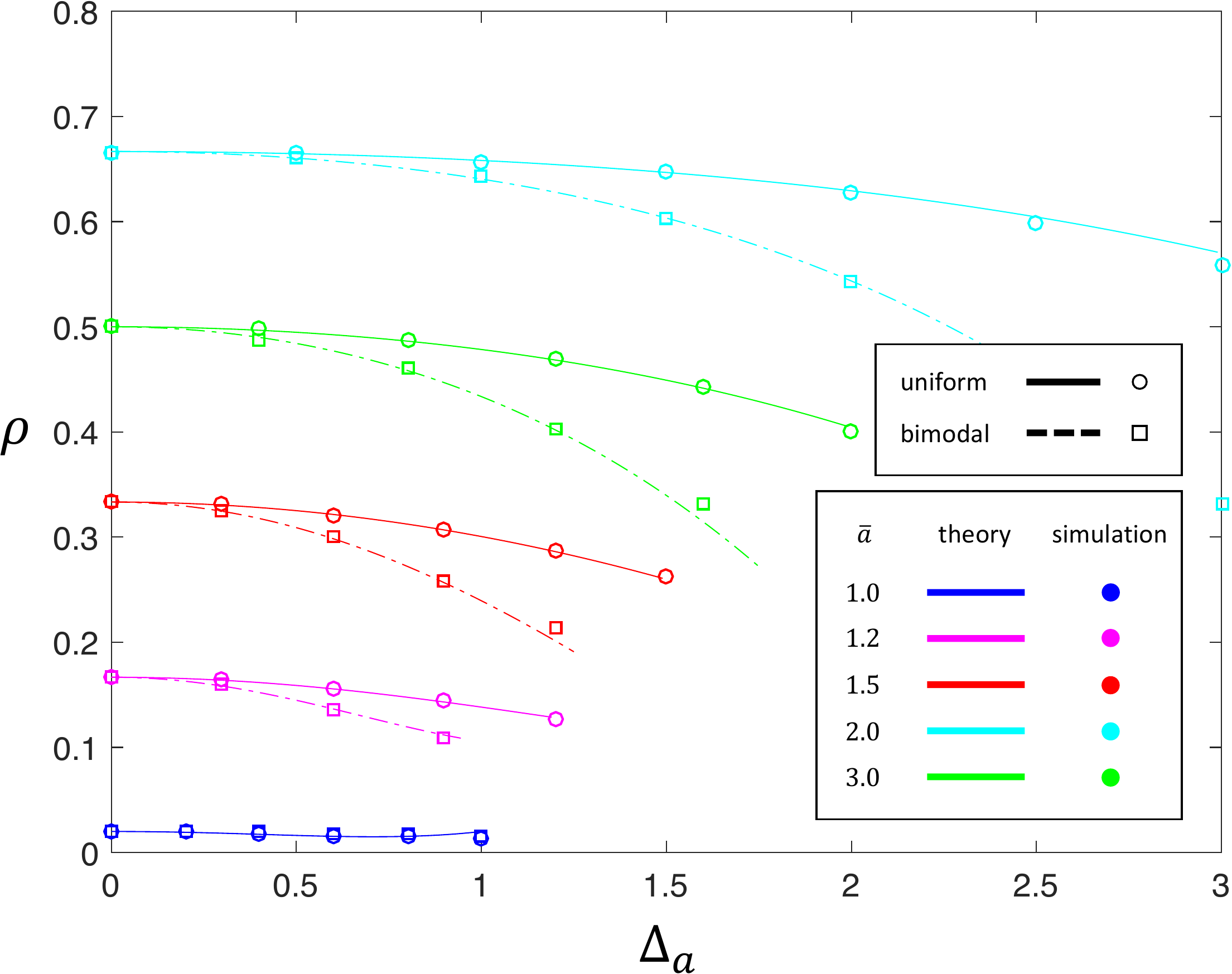}
\caption{Fixation probability of the mutant type, $A$, as a function of (half) the width of the mutant fitness distribution, $\Delta_{a}$. The fitness values for the mutant and resident are uniformly distributed on $\left[\overline{a}-\Delta_{a},\overline{a}+\Delta_{a}\right]$ and $\left[\overline{b}-\Delta_{b},\overline{b}+\Delta_{b}\right]$, respectively (solid line). Similarly, for a bimodal distribution, the fitness values for the mutant and resident are $\overline{a}-\Delta_{a}$ or $\overline{a}+\Delta_{a}$ and $\overline{b}-\Delta_{b}$ and $\overline{b}+\Delta_{b}$, respectively, each with probability $1/2$ (dashed lines). These values, $\Delta_{a}$ and $\Delta_{b}$, are measures of mutant and resident heterogeneity, respectively. The population size is $N=50$ and the solid/dashed lines indicate the analytical predictions from \textbf{Eq. \ref{eq:taylor_heterogeneity}}. As $\Delta_{a}$ grows, a beneficial mutant's fixation probability decreases. However, this fixation probability does not change as $\Delta_{b}$ varies (not shown in the figure).\label{fig:suppressedSelection}}
\end{figure}

\subsection{Effects on selection}
For fixed $f$ and $g$ with $\overline{a}>\overline{b}$, consider the function
\begin{align}
\psi &: \left[0,\infty\right) \longrightarrow \left[0,\infty\right) \nonumber \\
&: a \longmapsto \frac{a}{\overline{b}+a\rho_{A}^{\infty}\left(f,g\right)} .
\end{align}
Since $\psi$ is strictly concave whenever $\overline{a}>\overline{b}$, it follows from Jensen's inequality that
\begin{align}
1 &= \mathbb{E}_{f}\left[\psi\right] \leqslant \psi\left(\overline{a}\right) = \frac{\overline{a}}{\overline{b}+\overline{a}\rho_{A}^{\infty}\left(f,g\right)} ,
\end{align}
with equality if and only if there is no mutant heterogeneity (i.e. $f\left(\overline{a}\right) =1$). Therefore, if $\overline{a}>\overline{b}$, then $\rho_{A}^{\infty}\left(\overline{a},g\right) =1-\overline{b}/\overline{a}$, and we see that $\rho_{A}^{\infty}\left(f,g\right)\leqslant 1-\overline{b}/\overline{a}=\rho_{A}^{\infty}\left(\overline{a},g\right)$ with equality if and only if $f\left(\overline{a}\right) =1$. Thus, heterogeneity in the fitness of an advantageous mutant decreases its fixation probability (\fig{suppressedSelection}).

\subsection{Moment expansion of fixation probability}
Here, we discuss expansions for the fixation probability in the limit of weak heterogeneity. Let $f'$ and $g'$ be mass functions on $\mathbb{R}$, supported on the points $a_{1}',\dots ,a_{m}'\in\mathbb{R}$ and $b_{1}',\dots ,b_{m}'\in\mathbb{R}$, respectively. Suppose that $\mathbb{E}_{f'}\left[a'\right] =\mathbb{E}_{g'}\left[b'\right] =0$ (where, again, $\mathbb{E}_{f'}\left[a'\right]$ and $\mathbb{E}_{g'}\left[b'\right]$ denote the mean values of the random variables distributed according to $f'$ and $g'$, respectively). For $0<\varepsilon\ll 1$ and fixed $\overline{a},\overline{b}>0$ with $\overline{a}>\overline{b}$, consider the mass functions
\begin{subequations}
\begin{align}
f^{\left(\varepsilon\right)}\left(a\right) &\coloneqq f'\left(\frac{a-\overline{a}}{\varepsilon}\right) ; \\
g^{\left(\varepsilon\right)}\left(b\right) &\coloneqq g'\left(\frac{b-\overline{b}}{\varepsilon}\right) .
\end{align}
\end{subequations}
These functions are supported on the points $\left\{\overline{a}+\varepsilon a_{i}'\right\}_{i=1}^{m}$ and $\left\{\overline{b}+\varepsilon b_{i}'\right\}_{i=1}^{m}$, respectively.

Consider the series expansion of $\rho_{A}^{\infty}\left(f^{\left(\varepsilon\right)},g^{\left(\varepsilon\right)}\right)$ in terms of $\varepsilon$,
\begin{align}
\rho_{A}^{\infty}\left(f^{\left(\varepsilon\right)},g^{\left(\varepsilon\right)}\right) &= c_{0} + c_{1}\varepsilon + c_{2} \varepsilon^{2}+ c_{3}\varepsilon^{3}+c_{4}\varepsilon^{4} + \mathcal{O}\left(\varepsilon^{5}\right) .
\end{align}
We can solve for $c_{0}, c_{1},\dots ,c_{4}$ using a perturbative expansion of \eq{integralEquationB},
\begin{align}
1 &= \mathbb{E}_{f'}\left[ \frac{\left(\overline{a}+\varepsilon a'\right)}{\overline{b}+\left(\overline{a}+\varepsilon a'\right)\rho_{A}^{\infty}\left(f^{\left(\varepsilon\right)},g^{\left(\varepsilon\right)}\right)} \right] ,
\end{align}
and matching the coefficients for different powers of $\varepsilon$ up to $\varepsilon^{4}$. Since $\mathbb{E}_{f'}\left[a'\right] =0$, we see that
\begin{subequations}
\begin{align}
c_{0} &= 1-\frac{\overline{b}}{\overline{a}} ; \\
c_{1} &= 0 ; \\
c_{2} &= -\left(1-\frac{\overline{b}}{\overline{a}}\right)\left(\frac{\overline{b}}{\overline{a}^{3}}\right)\mathbb{E}_{f'}\left[\left(a'\right)^{2}\right] ; \\
c_{3} &= \left(1-\frac{\overline{b}}{\overline{a}}\right)\left(\frac{\overline{b}\left(\overline{a}-\overline{b}\right)}{\overline{a}^{5}}\right)\mathbb{E}_{f'}\left[\left(a'\right)^{3}\right] ; \\
c_{4} &= -\left(1-\frac{\overline{b}}{\overline{a}}\right)\left\{ \left(\frac{\overline{b}^{2}\left(\overline{a}-2\overline{b}\right)}{\overline{a}^{7}}\right)\mathbb{E}_{f'}\left[\left(a'\right)^{2}\right]^{2}-\left(\frac{\overline{b}\left(\overline{a}-\overline{b}\right)^{2}}{\overline{a}^{7}}\right)\mathbb{E}_{f'}\left[\left(a'\right)^{4}\right] \right\} .
\end{align}
\end{subequations}
Therefore, using the fact that $\rho_{A}^{\infty}\left(\overline{a},\overline{b}\right) =1-\overline{b}/\overline{a}$, we have
\begin{align}
\rho_{A}^{\infty}\left(f^{\left(\varepsilon\right)},g^{\left(\varepsilon\right)}\right) &\approx \rho_{A}^{\infty}\left(\overline{a},\overline{b}\right)\Bigg\{ 1 - \left(\frac{\overline{b}}{\overline{a}^{3}}\right)\mathbb{E}_{f'}\left[\left(a'\right)^{2}\right]\varepsilon^{2} + \left(\frac{\overline{b}\left(\overline{a}-\overline{b}\right)}{\overline{a}^{5}}\right)\mathbb{E}_{f'}\left[\left(a'\right)^{3}\right]\varepsilon^{3} \nonumber \\
&\quad\quad  -\left(\frac{\overline{b}^{2}\left(\overline{a}-2\overline{b}\right)}{\overline{a}^{7}}\right)\mathbb{E}_{f'}\left[\left(a'\right)^{2}\right]^{2}\varepsilon^{4}-\left(\frac{\overline{b}\left(\overline{a}-\overline{b}\right)^{2}}{\overline{a}^{7}}\right)\mathbb{E}_{f'}\left[\left(a'\right)^{4}\right]\varepsilon^{4} \Bigg\} .\label{sieq:taylorExpansion}
\end{align}
For symmetric distributions, the odd moments cancel, and this expansion can be simplified even further.

If $\rho_{A}^{\infty}\left(\overline{a},\overline{b}\right)$ is the fixation probability in the uniform (homogeneous) system, then it follows that one can approximate a mutant's fixation probability in the heterogeneous model using the expansion
\begin{align}
\rho_{A}^{\infty}\left(f,g\right) &\approx \rho_{A}^{\infty}\left(\overline{a},\overline{b}\right)\Bigg\{ 1 - \left(\frac{\overline{b}}{\overline{a}^{3}}\right)\mathbb{E}_{f}\left[\left(a-\overline{a}\right)^{2}\right] + \left(\frac{\overline{b}\left(\overline{a}-\overline{b}\right)}{\overline{a}^{5}}\right)\mathbb{E}_{f}\left[\left(a-\overline{a}\right)^{3}\right] \nonumber \\
&\quad\quad  -\left(\frac{\overline{b}^{2}\left(\overline{a}-2\overline{b}\right)}{\overline{a}^{7}}\right)\mathbb{E}_{f}\left[\left(a-\overline{a}\right)^{2}\right]^{2}-\left(\frac{\overline{b}\left(\overline{a}-\overline{b}\right)^{2}}{\overline{a}^{7}}\right)\mathbb{E}_{f}\left[\left(a-\overline{a}\right)^{4}\right] \Bigg\} . \label{eq:taylor_heterogeneity}
\end{align}
\textbf{Fig. \ref{fig:expansion}} demonstrates that this expansion is in excellent agreement with the simulation data.

\begin{figure}
\centering
\includegraphics[width=0.8\textwidth]{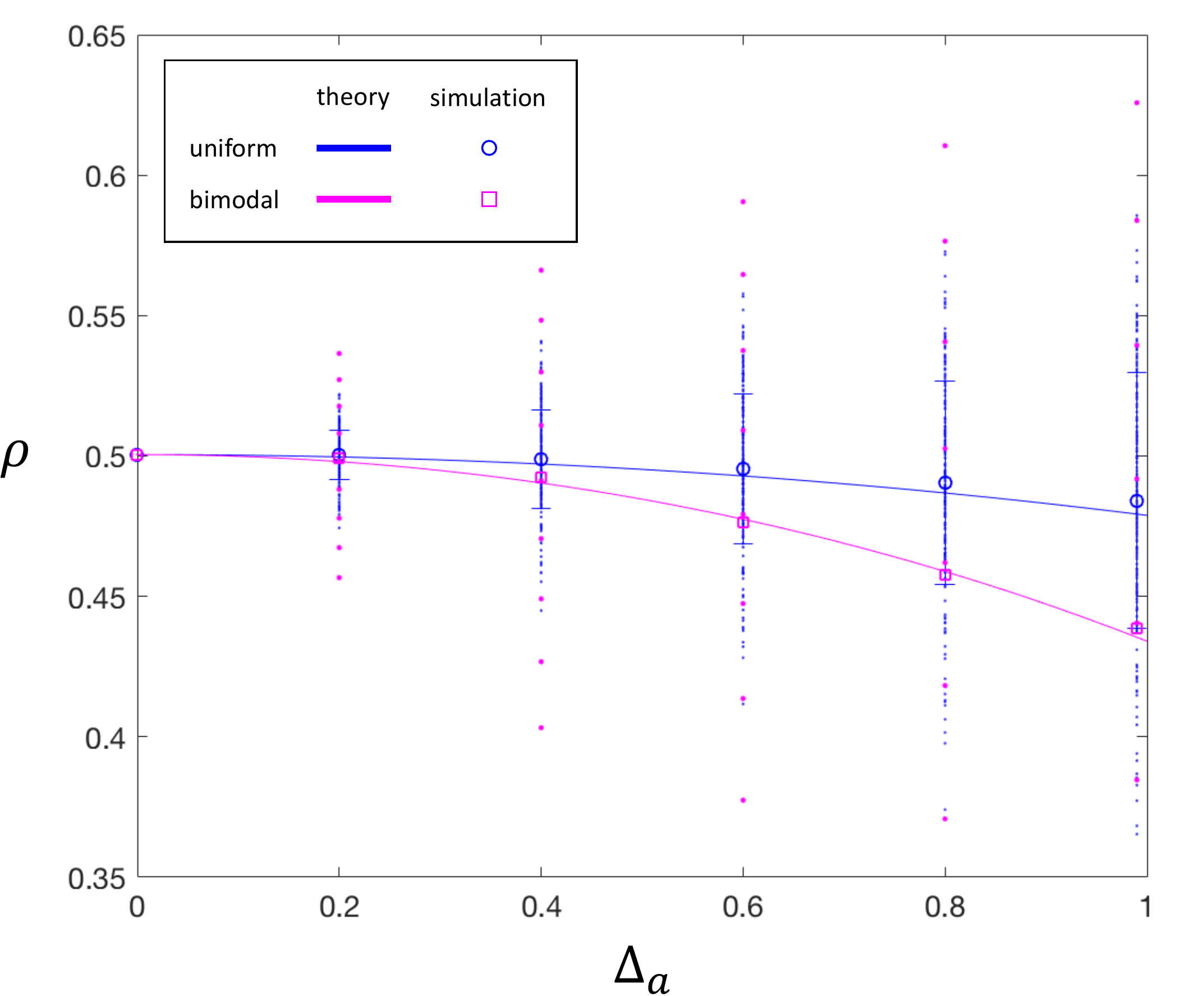}
\caption{The fixation probability of a randomly-placed mutant in a heterogeneous environment as function of mutant fitness width, $\Delta_{a}$. Fixation probabilities for each random configuration are derived from a given probability distribution with mean $\bar{a}$ and width $\Delta_{a}$. A bimodal distribution is shown in magenta and a uniform distribution is shown in blue. Small dots indicate the fixation probability in each random environmental configuration. (For each value $\Delta_{a}$, approximately 200 configurations are generated with $N=10$, $\bar{a}=2$, and $\bar{b}=1$.) The averaged fixation probability, depicted by a circle or a square, is in excellent agreement with the analytical results (\textbf{Eq.~\ref{eq:taylor_heterogeneity}}). For simplicity, fitness heterogeneity is assumed to apply only to mutants.\label{fig:expansion}}
\end{figure}

In \ref{sec:appendixB}, we show that altering the dispersal patterns can enhance this suppression effect. In other words, if an offspring can replace only certain individuals (instead of any other member of the population), then heterogeneity in mutant fitness further suppresses a rare mutant's fixation probability. In the case of a cycle with a spatially-periodic fitness distribution, the fixation probability approaches zero when the heterogeneity in mutant fitness approaches its maximal values (see \textbf{Fig.~\ref{fig:cycleSuppression}} in \ref{sec:appendixB}).

\section{Heterogeneity in resident fitness}
Although environmental heterogeneity of the resident is irrelevant when the population size is sufficiently large, it can have an effect on fixation probability for small population sizes. In most cases, this effect (which is of order $1/N$) can be ignored, but we observe that for small population sizes, and in particular near neutrality ($\overline{a}=\overline{b}$), heterogeneity in resident fitness values can amplify a mutant's fixation probability. One example of this amplification effect is presented in \textbf{Fig. \ref{fig:amplifiedSelection}}, where $\overline{a}$ is close to $1$ and $\mathbf{b}$ is distributed uniformly on $\left[\overline{b}-\Delta_{b},\overline{b}+\Delta_{b}\right]$, where $\overline{b}=1$. A second, bimodal distribution is also tested, with fitness values randomly chosen from two values, $\overline{b}-\Delta_{b}$ or $\overline{b}+\Delta_{b}$. In both cases, we observe that fixation probability is increased for near-neutral mutants. However, fixation probability is increased for both on-average beneficial and on-average deleterious mutations, which indicates that the mechanism of amplification is somewhat different from that of an amplifier of selection on evolutionary graphs (for example, a star graph). We also varied \textit{both} mutant and resident fitness; the heat map in \textbf{Fig. \ref{fig:heatmap}} summarizes the effects on fixation probability.

\begin{figure}
\centering
\includegraphics[width=0.8\textwidth]{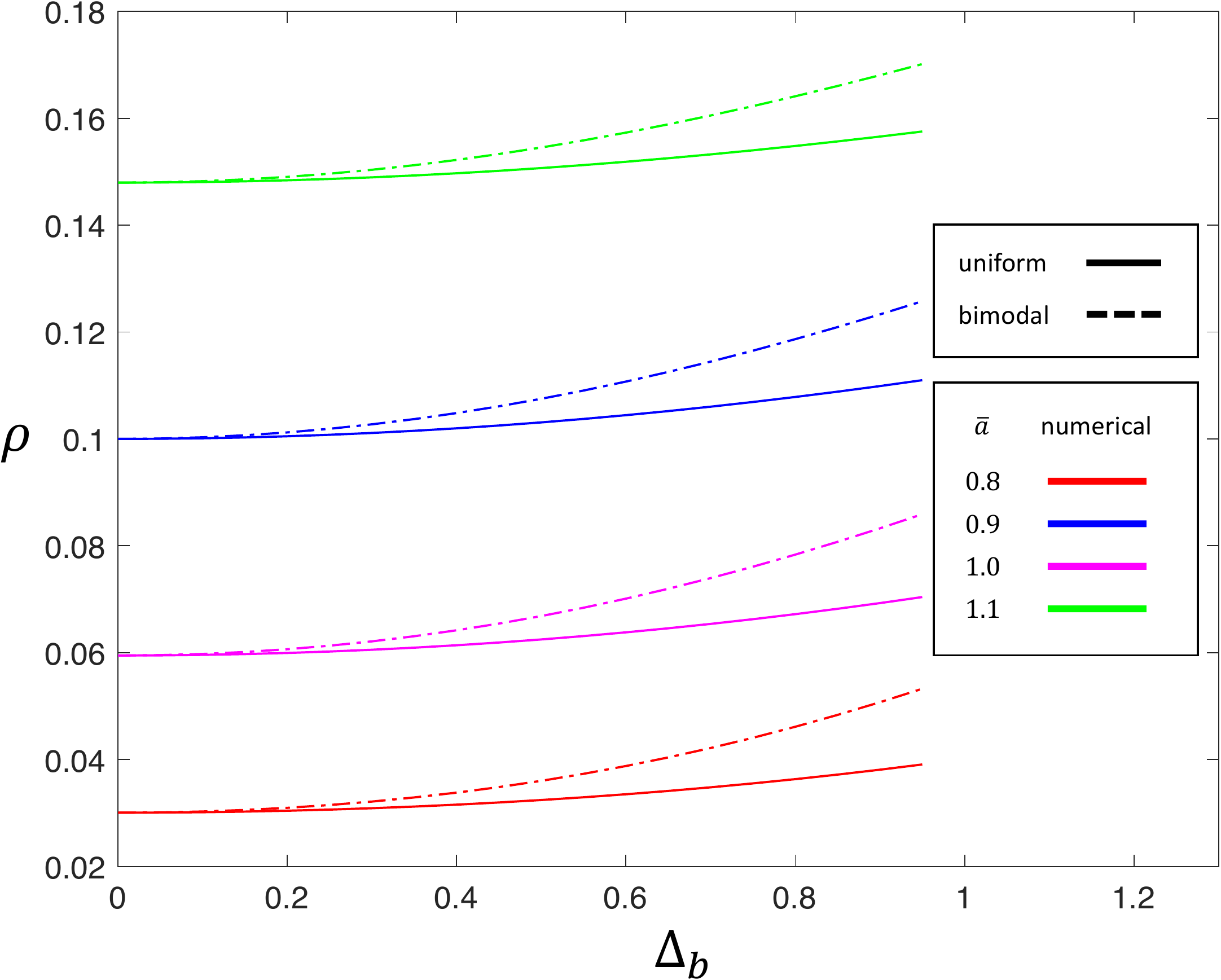}
\caption{Fixation probability of the mutant type, $A$, as a function of (half) the width of the resident fitness distribution, $\Delta_{b}$. The fitness values for the resident are uniformly distributed on $\left[\overline{b}-\Delta_{b},\overline{b}+\Delta_{b}\right]$ (solid line), where $\overline{b}=1$. Similarly, for a bimodal distribution, the fitness values for the resident are either $\overline{b}-\Delta_{b}$ or $\overline{b}+\Delta_{b}$, each chosen with probability $1/2$ (dashed lines). The population size is $N=10$, and $\overline{a}=0.8,0.9,1.0$ and $1.1$ (without any mutant fitness heterogeneity). The results are obtained from exact solutions of the Kolmogorov equation for the fixation probability. As $\Delta_{b}$ grows, a near-neutral mutant's fixation probability increases, consistent with amplification.\label{fig:amplifiedSelection}}
\end{figure}

\begin{figure}
\centering
\includegraphics[width=0.8\textwidth]{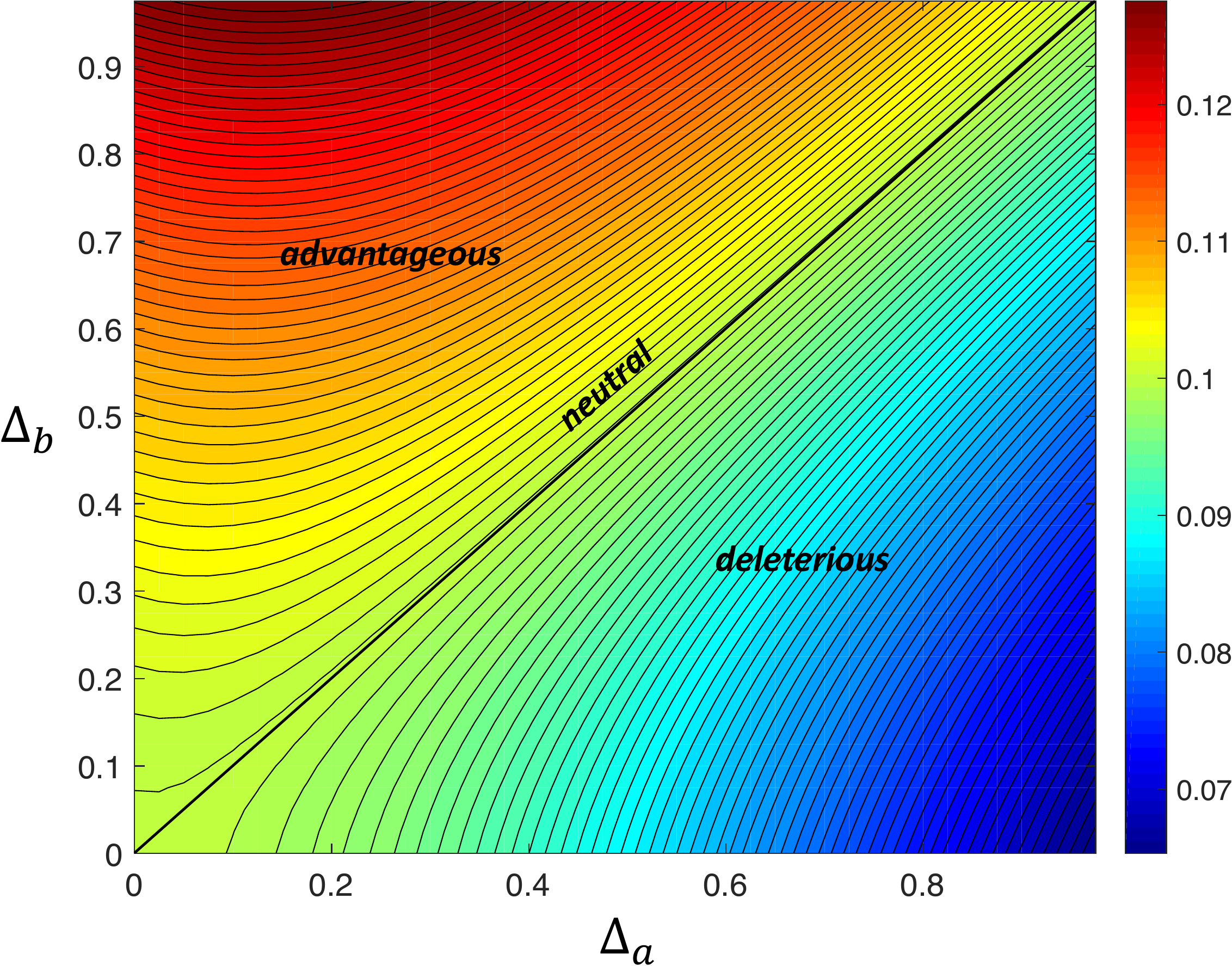}
\caption{Heat map for the fixation probability of the mutant type, $A$, as a function of (half) the width of the mutant fitness distribution, $\Delta_{a}$, and that of the resident fitness distribution, $\Delta_{b}$. The fitness values for the mutant and resident are uniformly distributed on $\left[\overline{a}-\Delta_{a},\overline{a}+\Delta_{a}\right]$ and $\left[\overline{b}-\Delta_{b},\overline{b}+\Delta_{b}\right]$, respectively, where $\overline{a}=\overline{b}=1$. The population size is $N=10$, and the results are obtained from numerical solutions to the Kolmogorov equation.\label{fig:heatmap}}
\end{figure}

\textbf{Fig.~\ref{fig:criticalN}} illustrates how these amplification effects change with population size, $N$. Once again, we show in \ref{sec:appendixB} that non-well-mixed dispersal patterns can further enhance the amplifying effects of heterogeneity in resident fitness. In the case of a cycle with spatially-periodic fitness values, an increase in the standard deviation of resident fitness leads to an even more noticeable increase in fixation probability (see \textbf{Fig.~\ref{fig:cycleAmplification}} in \ref{sec:appendixB}).

\begin{figure}
\centering
\includegraphics[width=0.8\textwidth]{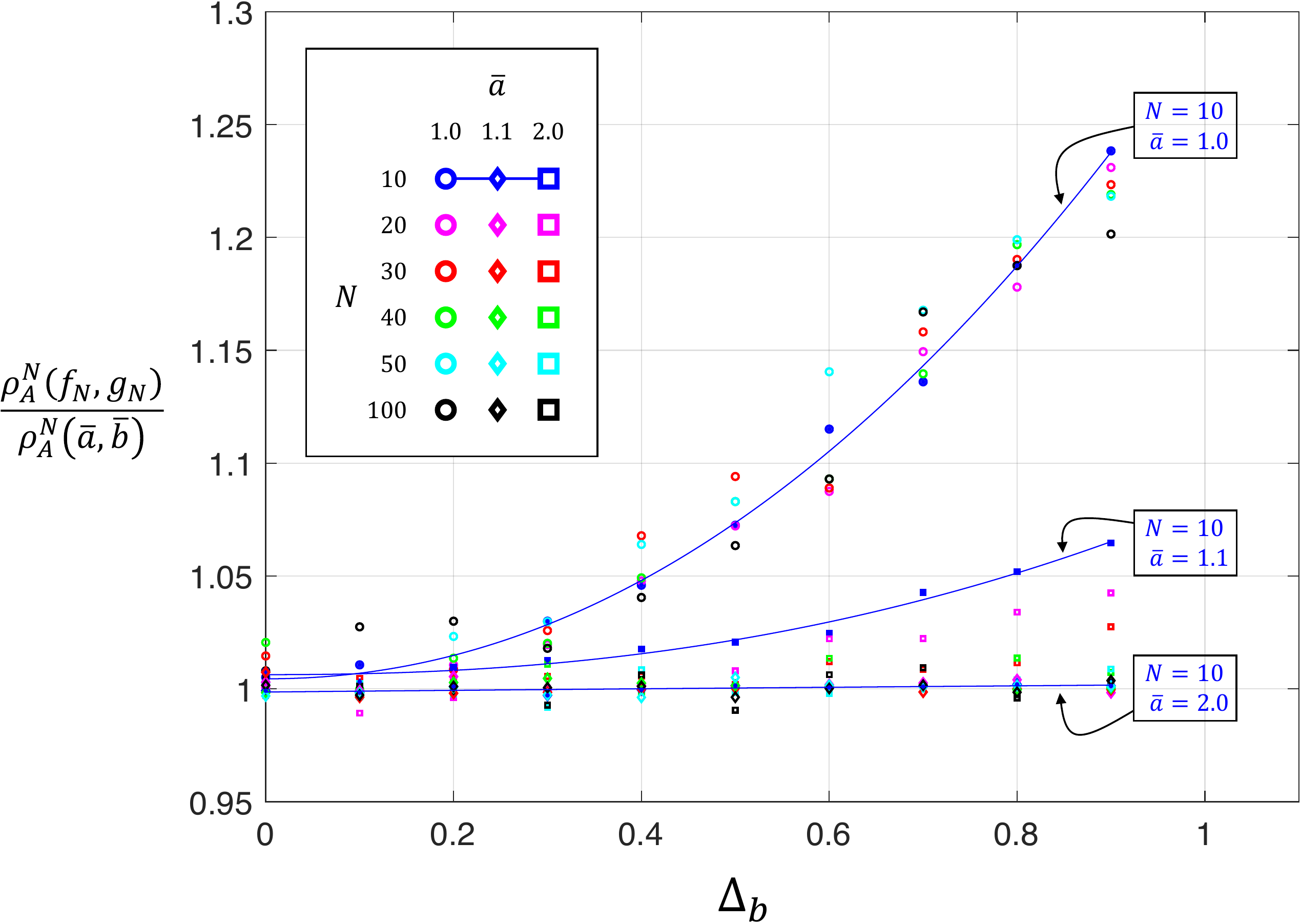}
\caption{Fixation probability for various population sizes, $N$, and mutant-fitness averages, $\bar{a}$. The heterogeneity is on resident fitness, using a uniform distribution function, i.e. $1-\Delta_{b} \leqslant b_{i} \leqslant 1+\Delta_{b}$. The fixation probability is normalized to that of uniform population with $\Delta_{b}=0$. The shapes (circle/diamond/square) indicate simulation results for various $\bar{a}$, and solid lines are interpolations for $N=10$ and $\bar{a}\in\left\{1,1.1,2\right\}$.\label{fig:criticalN}}
\end{figure}

\section{Discussion}
The Moran process has been studied extensively in structured populations, but spatial structure in this context usually pertains to the dispersal patterns of offspring following reproduction \citep{lieberman:Nature:2005,broom:JSTP:2011,frean:PRSB:2013,diaz:PRSA:2013,adlam:SR:2014,maciejewski:PLoSCB:2014,adlam:PRSA:2015,jamiesonlane:JTB:2015}. Other models use two graphs, with one ``interaction" graph pertaining to the payoffs that determine fitness and one ``dispersal" graph determining the propagation of offspring \citep{ohtsuki:PRL:2007,taylor:Nature:2007,ohtsuki:JTB:2007,pacheco:PLoSCB:2009,debarre:NC:2014}. The model of heterogeneity considered here is similar to these two-graph models since it allows for an environment-structured population yet has independent dispersal patterns. However, the structure of the environments cannot be captured by the same kind of interaction graph typically used in evolutionary game theory. Instead, the environments can be modeled by coloring the nodes of the dispersal graph, with one color for each distinct environment. The fitness of an individual is then determined by both the node's color and the individual's type. We discuss briefly in \ref{sec:appendixB} the dynamics on a cycle, which is a linear, periodic dispersal structure.

In heterogeneous environments, we find that there is a notable asymmetry between the mutant and resident types. Any variation in mutant fitness acts as a suppressor of selection. In particular, mutant heterogeneity decreases the fixation probability of beneficial mutants and increases the fixation probability of deleterious mutants. Resident heterogeneity, on the other hand, has no effect on a mutant's fixation probability in large populations and can even amplify it in small populations. Our finding differs from what is seen in processes with dispersal heterogeneity, which can amplify or suppress selection but need not do either \citep{adlam:PRSA:2015,hindersin:PLOSCB:2015,pavlogiannis:SR:2017}.

While the neutrality condition admits a simple interpretation when the population is large (i.e. the types have the same expected fitness; $\overline{a}=\overline{b}$), we do not expect this condition to be quite as intuitive for smaller population sizes. For smaller $N$, stochastic effects are stronger, and the neutrality condition is complicated by the interplay between natural selection and random drift; in large populations, selection becomes the primary effect. Even when $N=3$, we have seen that the neutrality condition is already quite complex.

Other kinds of fitness averages also arise in studies of environmental heterogeneity. In a two-allele model with ecological variation, the condition for the maintenance of a protected polymorphism is stated in terms of the harmonic mean of the fitness values \citep{levene:AN:1953}. If heterogeneity is temporal rather than environmental \citep{cvijovic:PNAS:2015}, then the mean in this condition is geometric \citep{haldane:JG:1963}. The approach we take here is somewhat different from these studies because we are focused instead on the contrast between two types under environmental heterogeneity. Furthermore, we treat a haploid Moran model, which has not been studied as extensively as diploid models with random mating--at least with regard to environmental fitness heterogeneity.

Heterogeneity, in its many and varied forms, is commonplace in evolving populations. Our focus here is on environmental fitness heterogeneity that can arise, for example, from spatial fluctuations in the availability of resources. Although mutant heterogeneity always suppresses selection and resident heterogeneity can amplify selection, it would be interesting to understand its interaction with other asymmetries such as those induced by spatial structure. In particular, how a combination of fitness and dispersal heterogeneity influences selection is poorly understood and represents an interesting topic for future research.

\setcounter{section}{0}
\renewcommand{\thesection}{Appendix~\Alph{section}}
\renewcommand{\thesubsection}{\Alph{section}.\arabic{subsection}}
\renewcommand{\thesubsubsection}{\Alph{section}.\arabic{subsection}.\arabic{subsubsection}}

\section{Fixation probabilities in heterogeneous environments}\label{sec:appendixA}
In this section, we establish an asymptotic formula for fixation probabilities in a heterogeneous environment. The population consists of $m$ different environments, whose only role is to determine reproductive fitness of the two types, $A$ and $B$. In environment $i$, which contains $N_{i}$ individuals, the $A$-type (resp. $B$-type) has fitness $a_{i}$ (resp. $b_{i}$), where $a_{i},b_{i}>0$. Once the fitness of each individual is determined, the process is updated as described in \S\ref{sec:model} via a Moran process in an unstructured population of size $N=N_{1}+\cdots +N_{m}$.

\subsection{State space and transition probabilities}
When the environment influences fitness, there are two possible notions of ``state space." One could simply track the trait of every individual in the population, which would result in the ``full" state space, $\left\{A,B\right\}^{N}$. However, since individuals within the same environment are indistinguishable, we instead use the ``reduced" state space, $S\coloneqq\left\{0,1,\dots ,N_{1}\right\}\times\cdots\times\left\{0,1,\dots ,N_{m}\right\}$. An element $\mathbf{n}=\left(n_{1},\dots ,n_{m}\right)\in S$ indicates the state in which there are $n_{i}$ individuals of type $A$ in environment $i$ for $i=1,\dots ,m$. For $\mathbf{n}\in S$, we denote the overall number of $A$-type individuals in $\mathbf{n}$ by $\left|\mathbf{n}\right|\coloneqq n_{1}+\cdots +n_{m}$.

Since we discuss what happens to fixation probabilities as the population size grows, we need a way to parametrize the population by only its size, $N$. Therefore, we assume that there are $m$ environments with $A$-fitness given by $a_{1},\dots ,a_{m}$ and $B$-fitness given by $b_{1},\dots ,b_{m}$. The size of environment $i$, $N_{i}$, is a function of $N$ with $N_{1}\left(N\right) +\cdots +N_{m}\left(N\right) =N$ for every $N\geqslant 1$. We assume that there exist $p_{1},\dots ,p_{m}\in\left(0,1\right)$ with
\begin{align}\label{sieq:patchFraction}
\lim_{N\rightarrow\infty}\frac{N_{i}\left(N\right)}{N} &= p_{i}	
\end{align}
for $i=1,\dots ,m$; one can think of environment $i$ as constituting a fixed, non-zero fraction, $p_{i}$, of the population. Since the dispersal structure is the same as that of an unstructured population, $N$ then completely specifies the population structure and the nature of fitness heterogeneity. This approach involves choosing a sequence of environment-structured populations, each with exactly $m$ environments, which is similar to how one uses sequences of populations to define a general notion of an amplifier of selection \citep{adlam:PRSA:2015}.

For $i=1,\dots ,m$, let $P_{i}^{+}\left(\mathbf{n}\right)$ (resp. $P_{i}^{-}\left(\mathbf{n}\right)$) be the probability that the number of mutants in environment $i$ goes up (resp. down) in the next update step, given that the current state is $\mathbf{n}$. In other words, $P_{i}^{+}\left(\mathbf{n}\right)$ is the probability that $n_{i}$ becomes $n_{i}+1$, and $P_{i}^{-}\left(\mathbf{n}\right)$ is the probability that $n_{i}$ becomes $n_{i}-1$. By the definition of the Moran process,
\begin{subequations}\label{sieq:transitions}
\begin{align}
P_{i}^{-}\left(\mathbf{n}\right) &= \left(\frac{\sum_{j=1}^{m}b_{j}\left(N_{j}-n_{j}\right)}{\sum_{j=1}^{m}\Big(a_{j}n_{j}+b_{j}\left(N_{j}-n_{j}\right)\Big)}\right) \left(\frac{n_{i}}{N}\right) ; \\
P_{i}^{+}\left(\mathbf{n}\right) &= \left(\frac{\sum_{j=1}^{m}a_{j}n_{j}}{\sum_{j=1}^{m}\Big(a_{j}n_{j}+b_{j}\left(N_{j}-n_{j}\right)\Big)}\right) \left(\frac{N_{i}-n_{i}}{N}\right) .
\end{align}
\end{subequations}
We denote by $\left\{X_{t}\right\}_{t\geqslant 0}$ the discrete-time Markov chain on $S$ generated by these transition probabilities.

However, in analyzing fixation probabilities, we may instead consider the chain $\left\{Y_{t}\right\}_{t\geqslant 0}$ on $S$ in which, for $i=1,\dots ,m$ and $\mathbf{n}\in S$ with $0<\left|\mathbf{n}\right| <N$, the transition probabilities are defined by the respective mutant-loss and mutant-gain probabilities,
\begin{subequations}\label{sieq:scaledTransitions}
\begin{align}
Q_{i}^{-}\left(\mathbf{n}\right) &\coloneqq \frac{P_{i}^{-}\left(\mathbf{n}\right)}{\sum_{j=1}^{m}\left(P_{j}^{-}\left(\mathbf{n}\right) +P_{j}^{+}\left(\mathbf{n}\right)\right)} \nonumber \\
&= \left(\frac{\left|\mathbf{n}\right|\sum_{j=1}^{m}b_{j}\left(N_{j}-n_{j}\right)}{\left|\mathbf{n}\right|\sum_{j=1}^{m}b_{j}\left(N_{j}-n_{j}\right) +\left(N-\left|\mathbf{n}\right|\right)\sum_{j=1}^{m}a_{j}n_{j}}\right)\left(\frac{n_{i}}{\left|\mathbf{n}\right|}\right) ; \\
Q_{i}^{+}\left(\mathbf{n}\right) &\coloneqq \frac{P_{i}^{+}\left(\mathbf{n}\right)}{\sum_{j=1}^{m}\left(P_{j}^{-}\left(\mathbf{n}\right) +P_{j}^{+}\left(\mathbf{n}\right)\right)} \nonumber \\
&= \left(\frac{\left(N-\left|\mathbf{n}\right|\right)\sum_{j=1}^{m}a_{j}n_{j}}{\left|\mathbf{n}\right|\sum_{j=1}^{m}b_{j}\left(N_{j}-n_{j}\right) +\left(N-\left|\mathbf{n}\right|\right)\sum_{j=1}^{m}a_{j}n_{j}}\right)\left(\frac{N_{i}-n_{i}}{N-\left|\mathbf{n}\right|}\right) .
\end{align}
\end{subequations}
The two monomorphic states of this chain, $\mathbf{n}=N$ and $\mathbf{n}=0$ (which we denote by $\mathbf{A}$ and $\mathbf{B}$, respectively), are absorbing. Note that this chain can still be described in terms of births and replacements. Specifically, a resident birth occurs with probability $\frac{\left|\mathbf{n}\right|\sum_{j=1}^{m}b_{j}\left(N_{j}-n_{j}\right)}{\left|\mathbf{n}\right|\sum_{j=1}^{m}b_{j}\left(N_{j}-n_{j}\right) +\left(N-\left|\mathbf{n}\right|\right)\sum_{j=1}^{m}a_{j}n_{j}}$, and the offspring replaces a mutant in environment $i$ with probability $\frac{n_{i}}{\left|\mathbf{n}\right|}$; a mutant birth occurs with probability $\frac{\left(N-\left|\mathbf{n}\right|\right)\sum_{j=1}^{m}a_{j}n_{j}}{\left|\mathbf{n}\right|\sum_{j=1}^{m}b_{j}\left(N_{j}-n_{j}\right) +\left(N-\left|\mathbf{n}\right|\right)\sum_{j=1}^{m}a_{j}n_{j}}$, and the offspring replaces a resident in environment $i$ with probability $\frac{N_{i}-n_{i}}{N-\left|\mathbf{n}\right|}$.

That the fixation probabilities are the same in $\left\{X_{t}\right\}_{t\geqslant 0}$ and $\left\{Y_{t}\right\}_{t\geqslant 0}$ can be seen from their recurrence relations. Specifically, if $\mathbf{n}_{i}^{-}$ (resp. $\mathbf{n}_{i}^{+}$) denotes the state obtained from $\mathbf{n}$ by changing $n_{i}$ to $n_{i}+1$ (resp. $n_{i}-1$), and if $\rho_{\mathbf{n},\mathbf{A}}$ is the probability of reaching the all-$A$ state, $\mathbf{A}$, when the process starts in state $\mathbf{n}$, then
\begin{align}
\rho_{\mathbf{n},\mathbf{A}} &= \left(1-\sum_{i=1}^{m}P_{i}^{-}\left(\mathbf{n}\right) -\sum_{i=1}^{m}P_{i}^{+}\left(\mathbf{n}\right)\right)\rho_{\mathbf{n},\mathbf{A}} + \sum_{i=1}^{m}P_{i}^{-}\left(\mathbf{n}\right)\rho_{\mathbf{n}_{i}^{-},\mathbf{A}} + \sum_{i=1}^{m}P_{i}^{+}\left(\mathbf{n}\right)\rho_{\mathbf{n}_{i}^{+},\mathbf{A}} \nonumber \\
&\iff \rho_{\mathbf{n},\mathbf{A}} = \sum_{i=1}^{m}Q_{i}^{-}\left(\mathbf{n}\right)\rho_{\mathbf{n}_{i}^{-},\mathbf{A}} + \sum_{i=1}^{m}Q_{i}^{+}\left(\mathbf{n}\right)\rho_{\mathbf{n}_{i}^{+},\mathbf{A}}
\end{align}
\citep[see][]{kemeny:S:1960}. Therefore, in what follows we analyze the chain $\left\{Y_{t}\right\}_{t\geqslant 0}$ for simplicity.

\subsection{Limiting process}
From \textbf{Eq.~\ref{sieq:patchFraction}}, for every fixed $\mathbf{n}=\left(n_{1},\dots ,n_{m}\right)\in\left\{0,1,\dots\right\}^{m}$, there exists $N^{\ast}$ sufficiently large such that $n_{i}\leqslant N_{i}\left(N\right)$ for each $i=1,\dots ,m$ whenever $N\geqslant N^{\ast}$. The expected fitness of a randomly-placed individual of type $A$ (resp. $B$) in the limit is then $\overline{a}\coloneqq\lim_{N\rightarrow\infty}\frac{1}{N}\sum_{i=1}^{m}N_{i}a_{i}=\sum_{i=1}^{m}p_{i}a_{i}$ (resp. $\overline{b}\coloneqq\lim_{N\rightarrow\infty}\frac{1}{N}\sum_{i=1}^{m}N_{i}b_{i}=\sum_{i=1}^{m}p_{i}b_{i}$). Therefore, we have the following limits:
\begin{subequations}\label{sieq:Qtransitions}
\begin{align}
\widetilde{Q}_{i}^{-}\left(\mathbf{n}\right) &\coloneqq \lim_{N\rightarrow\infty} Q_{i}^{-}\left(\mathbf{n}\right) = \frac{n_{i}\overline{b}}{\left|\mathbf{n}\right|\overline{b}+\sum_{j=1}^{m}n_{j}a_{j}} ; \\
\widetilde{Q}_{i}^{+}\left(\mathbf{n}\right) &\coloneqq \lim_{N\rightarrow\infty} Q_{i}^{+}\left(\mathbf{n}\right) = \frac{p_{i}\sum_{j=1}^{m}n_{j}a_{j}}{\left|\mathbf{n}\right|\overline{b}+\sum_{j=1}^{m}n_{j}a_{j}} .
\end{align}
\end{subequations}
Of course, when $\mathbf{n}=\mathbf{0}$, we have $\widetilde{Q}_{i}^{\pm}\left(\mathbf{0}\right) =0$ (i.e. $\mathbf{0}$ is the only absorbing state of the process).

In other words, this limit defines a Markov chain, $\left\{\widetilde{Y}_{t}\right\}_{t\geqslant 0}$ on $\left\{0,1,\dots\right\}^{m}$ with transition probabilities given by $\widetilde{Q}_{i}^{\pm}\left(\mathbf{n}\right)$ for $\mathbf{n}\in\left\{0,1,\dots\right\}^{m}$. The probability of staying in the same state, $\mathbf{n}$, is $0$ in this chain, unless $\mathbf{n}=\mathbf{0}$ (which is an absorbing state). In fact, we can ignore $B$ entirely and think of the Markov chain defined by $\widetilde{Q}$ as giving transition probabilities in an $m$-type population of variable size, where, in state $\mathbf{n}$, the number of individuals of type $i$ in the population is $n_{i}$ and the size of the population is $\left|\mathbf{n}\right|$. When an individual of type $i$ gives birth, the offspring acquires type $j$ with probability $p_{j}$ (which, notably, is independent of $i$ and thus is the same for all birth events).

Fix $\mathbf{n}$ and $\ell >0$ with $\left|\mathbf{n}\right| <\ell$. Consider the problem of finding the probability, $\mathcal{E}_{\mathbf{n}}^{\ell}$, of hitting $\mathbf{0}$ (extinction) before hitting a state with at least $\ell$ individuals. The extinction probabilities satisfy the equation
\begin{align}\label{sieq:extinctionRecurrence}
\mathcal{E}_{\mathbf{n}}^{\ell} &= \sum_{i=1}^{m} \widetilde{Q}_{i}^{-}\left(\mathbf{n}\right) \mathcal{E}_{\mathbf{n}_{i}^{-}}^{\ell} + \sum_{i=1}^{m} \widetilde{Q}_{i}^{+}\left(\mathbf{n}\right) \mathcal{E}_{\mathbf{n}_{i}^{+}}^{\ell} ,
\end{align}
with boundary conditions $\mathcal{E}_{\mathbf{n}}^{\ell}=1$ if $\left|\mathbf{n}\right| =0$ and $\mathcal{E}_{\mathbf{n}}^{\ell}=0$ if $\left|\mathbf{n}\right| =\ell$.

For any $\gamma_{1},\dots ,\gamma_{m}\in\mathbb{R}$, let $\bm{\gamma}_{\mathbf{n}}\coloneqq\gamma_{1}^{n_{1}}\cdots\gamma_{m}^{n_{m}}$. Suppose that $\gamma_{1}^{\ast},\dots ,\gamma_{m}^{\ast}>0$ satisfy
\begin{align}\label{sieq:gammaMartingale}
\bm{\gamma}_{\mathbf{n}}^{\ast} &= \sum_{i=1}^{m} \widetilde{Q}_{i}^{-}\left(\mathbf{n}\right) \bm{\gamma}_{\mathbf{n}_{i}^{-}}^{\ast} + \sum_{i=1}^{m} \widetilde{Q}_{i}^{+}\left(\mathbf{n}\right) \bm{\gamma}_{\mathbf{n}_{i}^{+}}^{\ast} .
\end{align}
To find $\gamma_{1}^{\ast},\dots ,\gamma_{m}^{\ast}>0$ satisfying \textbf{Eq.~\ref{sieq:gammaMartingale}}, we first consider the case in which $\mathbf{n}=\mathbf{e}_{i}$ for some $i=1,\dots ,m$, where $\mathbf{e}_{i}$ denotes the state with $n_{i}=1$ and $n_{j}=0$ for $j\neq i$. With $\overline{\gamma^{\ast}}\coloneqq\sum_{i=1}^{m}p_{i}\gamma_{i}^{\ast}$, \textbf{Eq.~\ref{sieq:gammaMartingale}} reads
\begin{align}
\gamma_{i}^{\ast} &= \frac{\overline{b}}{\overline{b}+a_{i}} + \frac{a_{i}}{\overline{b}+a_{i}} \gamma_{i}^{\ast}\overline{\gamma^{\ast}} .
\end{align}
Solving for $\gamma_{i}^{\ast}$ then gives
\begin{align}\label{sieq:gammaInfty}
\gamma_{i}^{\ast} &= \frac{\overline{b}}{\overline{b}+a_{i}\left(1-\overline{\gamma^{\ast}}\right)} ,
\end{align}
where $\overline{\gamma^{\ast}}$ satisfies the equation
\begin{align}\label{sieq:gammaBarInfty}
1-\overline{\gamma^{\ast}} &= \left(1-\overline{\gamma^{\ast}}\right)\sum_{i=1}^{m} p_{i} \left(\frac{a_{i}}{\overline{b}+a_{i}\left(1-\overline{\gamma^{\ast}}\right)}\right) .
\end{align}

\begin{remark}
If $f\left(a\right)$ is the mass function defined by $f\left(a_{i}\right) =p_{i}$ for $i=1,\dots ,m$, with $f\left(a\right) =0$ whenever $a\neq a_{1},\dots ,a_{m}$, then the summation in \textbf{Eq.~\ref{sieq:gammaBarInfty}} is simply the expectation $\mathbb{E}_{f}\left[\frac{a}{\overline{b}+a\left(1-\overline{\gamma^{\ast}}\right)}\right]$.
\end{remark}

The following lemma characterizes the values of $\overline{\gamma^{\ast}}$ that satisfy \textbf{Eq.~\ref{sieq:gammaBarInfty}}:
\begin{lemma}\label{lem:overlineAB}
If $\overline{a}\leqslant\overline{b}$, then the only solution to \textbf{Eq.~\ref{sieq:gammaBarInfty}} in the interval $\left[0,1\right]$ is $\overline{\gamma^{\ast}}=1$. If $\overline{a}>\overline{b}$, then there are exactly two distinct solutions to \textbf{Eq.~\ref{sieq:gammaBarInfty}}: one at $\overline{\gamma^{\ast}}=1$ and another with $0<\overline{\gamma^{\ast}}<1$.
\end{lemma}
\begin{proof}
We first note that $\overline{\gamma^{\ast}}=1$ is always a solution to \textbf{Eq.~\ref{sieq:gammaBarInfty}}. Consider the change of variables $x\coloneqq 1-\overline{\gamma^{\ast}}$. The function $\xi\left(x\right)\coloneqq\sum_{i=1}^{m} p_{i} \left(\frac{a_{i}}{\overline{b}+a_{i}x}\right)$ is monotonically decreasing in $x$ with $\xi\left(0\right) =\overline{a}/\overline{b}$. As a result, when $\overline{a}\leqslant\overline{b}$, we have $\xi\left(0\right)\leqslant 1$ and $\xi\left(x\right) <1$ when $x>0$. Suppose now that $\overline{a}>\overline{b}$. Since the function $a\mapsto\frac{a}{\overline{b}+ax}$ is concave in $a$ whenever $a,x>0$, it follows from Jensen's inequality that $\xi\left(x\right)\leqslant\frac{\overline{a}}{\overline{b}+\overline{a}x}$ for every $x>0$. Therefore, $\xi\left(1\right)\leqslant\frac{\overline{a}}{\overline{b}+\overline{a}}<1$, and since $\xi$ is continuous in $x$ on $\left[0,\infty\right)$ with $\xi\left(0\right) >1$, there exists $x\in\left(0,1\right)$ for which $\xi\left(x\right) =1$ by virtue of the intermediate value theorem, which completes the proof.
\end{proof}

Although we chose $\gamma_{1}^{\ast},\dots ,\gamma_{m}^{\ast}$ that satisfy \textbf{Eq.~\ref{sieq:gammaMartingale}} for $\mathbf{n}=\mathbf{e}_{i}$, these values actually satisfy this equation for any $\mathbf{n}$. To see why, note first that $\bm{\gamma}_{\mathbf{n}_{i}^{+}}^{\ast}=\gamma_{i}^{\ast}\bm{\gamma}_{\mathbf{n}}^{\ast}=\gamma_{i}^{\ast}\gamma_{j}^{\ast}\bm{\gamma}_{\mathbf{n}_{j}^{-}}^{\ast}$ and $\overline{b}+a_{i}\gamma_{i}^{\ast}\overline{\gamma^{\ast}}=\left(\overline{b}+a_{i}\right)\gamma_{i}^{\ast}$ for every $i,j=1,\dots ,m$. Therefore,
\begin{align}
\sum_{i=1}^{m} n_{i}\overline{b} \bm{\gamma}_{\mathbf{n}_{i}^{-}}^{\ast} + \sum_{i=1}^{m} p_{i}\sum_{j=1}^{m}n_{j}a_{j} \bm{\gamma}_{\mathbf{n}_{i}^{+}}^{\ast} &= \sum_{i=1}^{m} n_{i}\overline{b} \bm{\gamma}_{\mathbf{n}_{i}^{-}}^{\ast} + \sum_{j=1}^{m}n_{j}a_{j}\sum_{i=1}^{m} p_{i} \bm{\gamma}_{\mathbf{n}_{i}^{+}}^{\ast} \nonumber \\
&= \sum_{i=1}^{m} n_{i}\overline{b} \bm{\gamma}_{\mathbf{n}_{i}^{-}}^{\ast} + \sum_{i=1}^{m}n_{i}a_{i}\gamma_{i}^{\ast}\overline{\gamma^{\ast}}\bm{\gamma}_{\mathbf{n}_{i}^{-}}^{\ast} \nonumber \\
&= \sum_{i=1}^{m}n_{i}\left(\overline{b} + a_{i}\gamma_{i}^{\ast}\overline{\gamma^{\ast}}\right)\bm{\gamma}_{\mathbf{n}_{i}^{-}}^{\ast} \nonumber \\
&= \sum_{i=1}^{m}n_{i}\left(\overline{b}+a_{i}\right)\gamma_{i}^{\ast}\bm{\gamma}_{\mathbf{n}_{i}^{-}}^{\ast} \nonumber \\
&= \left(\left|\mathbf{n}\right|\overline{b}+\sum_{j=1}^{m}a_{i}n_{i}\right)\bm{\gamma}_{\mathbf{n}}^{\ast} ,
\end{align}
which establishes \textbf{Eq.~\ref{sieq:gammaMartingale}}.

From \textbf{Eq.~\ref{sieq:gammaMartingale}}, we see that $\mathbb{E}\left[\bm{\gamma}_{\widetilde{Y}_{t+1}}^{\ast}\ :\ \widetilde{Y}_{t}=\mathbf{n}\right] =\bm{\gamma}_{\mathbf{n}}^{\ast}$ for every $\mathbf{n}$, meaning $\left\{\bm{\gamma}_{\widetilde{Y}_{t}}^{\ast}\right\}_{t\geqslant 0}$ is a Martingale with respect to $\left\{\widetilde{Y}_{t}\right\}_{t\geqslant 0}$. Consider the stopping time $\tau_{\ell}\coloneqq\min\left\{t\geqslant 0\ :\ \left|\widetilde{Y}_{t}\right| =\ell\right\}$, and let $\left\{\widetilde{Y}_{t}^{\tau_{\ell}}\right\}_{t\geqslant 0}$ be the stopped chain defined by $\widetilde{Y}_{t}^{\tau_{\ell}}\coloneqq\widetilde{Y}_{\min\left\{t,\tau_{\ell}\right\}}$. For any $\mathbf{n}$ with $\left|\mathbf{n}\right|\leqslant\ell$, it follows trivially from the Martingale property that $\mathbb{E}\left[\bm{\gamma}_{\widetilde{Y}_{t+1}^{\tau_{\ell}}}^{\ast}\ :\ \widetilde{Y}_{t}^{\tau_{\ell}}=\mathbf{n}\right] =\bm{\gamma}_{\mathbf{n}}^{\ast}$. Taking the limit of this equation as $t\rightarrow\infty$ \citep[see][]{monk:PRSA:2014}, we have
\begin{align}
\bm{\gamma}_{\mathbf{n}}^{\ast} &= \mathcal{E}_{\mathbf{n}}^{\ell} + \sum_{\substack{\mathbf{n}'\in\left\{0,1,\dots ,\right\}^{m} \\ \left|\mathbf{n}'\right| =\ell}}\mathbb{P}\left[ \widetilde{Y}_{\tau_{\ell}}=\mathbf{n}'\ :\ \widetilde{Y}_{0}=\mathbf{n} ,\ \tau_{\ell}<\infty \right]\bm{\gamma}_{\mathbf{n}'}^{\ast} .
\end{align}
Since $\sum_{\substack{\mathbf{n}'\in\left\{0,1,\dots ,\right\}^{m} \\ \left|\mathbf{n}'\right| =\ell}}\mathbb{P}\left[ \widetilde{Y}_{\tau_{\ell}}=\mathbf{n}'\ :\ \widetilde{Y}_{0}=\mathbf{n} ,\ \tau_{\ell}<\infty \right] =1-\mathcal{E}_{\mathbf{n}}^{\ell}$, we obtain the inequalities
\begin{align}
\left(1-\left(\min_{1\leqslant i\leqslant m}\gamma_{i}^{\ast}\right)^{\ell}\right)\mathcal{E}_{\mathbf{n}}^{\ell} + \left(\min_{1\leqslant i\leqslant m}\gamma_{i}^{\ast}\right)^{\ell} &\leqslant \bm{\gamma}_{\mathbf{n}}^{\ast} \leqslant \left(1-\left(\max_{1\leqslant i\leqslant m}\gamma_{i}^{\ast}\right)^{\ell}\right)\mathcal{E}_{\mathbf{n}}^{\ell} + \left(\max_{1\leqslant i\leqslant m}\gamma_{i}^{\ast}\right)^{\ell} . \label{sieq:Einequalities}
\end{align}
When $\overline{\gamma^{\ast}}<1$, we know that $\gamma_{i}^{\ast}<1$ for all $i=1,\dots ,m$ by \textbf{Eq.~\ref{sieq:gammaInfty}}, which gives the bound
\begin{align}
\mathcal{E}_{\mathbf{n}}^{\ell} &\geqslant \frac{\bm{\gamma}_{\mathbf{n}}^{\ast}-\left(\max_{1\leqslant i\leqslant m}\gamma_{i}^{\ast}\right)^{\ell}}{1-\left(\max_{1\leqslant i\leqslant m}\gamma_{i}^{\ast}\right)^{\ell}} . \label{sieq:calELowerBound}
\end{align}
Taking a sequence of fitness values for which $\overline{a}\downarrow\overline{b}$, the arguments of Lemma~\ref{lem:overlineAB} imply that $\overline{\gamma^{\ast}}\uparrow 1$. Moreover, taking the limit $\overline{\gamma^{\ast}}\uparrow 1$ in \textbf{Eq.~\ref{sieq:calELowerBound}} (using the expressions for $\gamma_{i}^{\ast}$ from \textbf{Eq.~\ref{sieq:gammaInfty}}), then gives
\begin{align}
\mathcal{E}_{\mathbf{n}}^{\ell} &\geqslant 1 - \frac{1}{\ell}\sum_{i=1}^{m}n_{i}\frac{a_{i}}{\min_{1\leqslant j\leqslant m}a_{j}} . \label{sieq:gammaOneInequality}
\end{align}
From the inequalities of \textbf{Eq.~\ref{sieq:Einequalities}} and \textbf{Eq.~\ref{sieq:gammaOneInequality}}, we thus have
\begin{align}\label{sieq:limitingE}
\lim_{\ell\rightarrow\infty}\mathcal{E}_{\mathbf{n}}^{\ell} &= 
\begin{cases}
\bm{\gamma}_{\mathbf{n}}^{\ast} & \gamma_{1}^{\ast},\dots ,\gamma_{m}^{\ast} < 1 , \\
1 & \gamma_{1}^{\ast},\dots ,\gamma_{m}^{\ast}\geqslant 1 .
\end{cases}
\end{align}
Therefore, it follows from Lemma~\ref{lem:overlineAB} and \textbf{Eq.~\ref{sieq:limitingE}} that
\begin{align}
\lim_{\ell\rightarrow\infty}\mathcal{E}_{\mathbf{n}}^{\ell} &= 
\begin{cases}
\left(\gamma_{1}^{\ast}\right)^{n_{1}}\cdots\left(\gamma_{m}^{\ast}\right)^{n_{m}} & \overline{a}>\overline{b} , \\
1 & \overline{a}\leqslant\overline{b} .
\end{cases}
\end{align}

Although we know the extinction probabilities in the limiting process, we cannot immediately conclude that these must coincide with the limit of the extinction probabilities in the Moran process. This situation is analogous to the use of branching processes approximations: while branching processes can be used to derive simple approximations of quantities in large populations \citep{nowak:BP:2006,wild:BMB:2010,leventhal:NC:2015}, one must also know that the use of such an approximation is valid for the process under consideration \citep{durrett:AAP:2009}. In the next section, we provide a sketch of how to find $\lim_{N\rightarrow\infty}\rho_{\mathbf{e}_{i},\mathbf{A}}^{N}$ using the extinction probabilities derived thus far.

\subsection{Large-population limit of fixation probability}
When the overall population size is $N$, let $\rho_{\mathbf{n},\ell}^{N}$ denote the probability of hitting a state with $\ell$ mutants when the process starts in state $\mathbf{n}$. To find $\rho_{\mathbf{e}_{i},\mathbf{A}}^{\infty}\coloneqq\lim_{N\rightarrow\infty}\rho_{\mathbf{e}_{i},\mathbf{A}}^{N}$, we first find $\lim_{\ell\rightarrow\infty}\lim_{N\rightarrow\infty}\rho_{\mathbf{e}_{i},\ell}^{N}$ and then argue that $\lim_{N\rightarrow\infty}\rho_{\mathbf{e}_{i},\mathbf{A}}^{N}=\lim_{\ell\rightarrow\infty}\lim_{N\rightarrow\infty}\rho_{\mathbf{e}_{i},\ell}^{N}$.

\begin{lemma}\label{lem:extinctionFixation}
For any $\ell$ and $\mathbf{n}$ with $\left|\mathbf{n}\right|\leqslant\ell$, $\lim_{N\rightarrow\infty}\rho_{\mathbf{n},\ell}^{N}$ exists and equals $1-\mathcal{E}_{\mathbf{n}}^{\ell}$.
\end{lemma}
\begin{proof}
In the chain $\left\{Y_{t}\right\}_{t\geqslant 0}$, consider the stopping time $\tau_{\ell}\coloneqq\min\left\{t\geqslant 0\ :\ \left| Y_{t}\right| =\ell\right\}$ and let $\left\{Y_{t}^{\tau_{\ell}}\right\}_{t\geqslant 0}$ be the stopped chain defined by $Y_{t}^{\tau_{\ell}}\coloneqq Y_{\min\left\{t,\tau_{\ell}\right\}}$. For every $N$, $\left\{Y_{t}^{\tau_{\ell}}\right\}_{t\geqslant 0}$ is defined on the finite state space $\left\{\mathbf{n}\ :\ \left|\mathbf{n}\right|\leqslant\ell\right\}$, which, importantly, is independent of $N$. For $\mathbf{n}\in S$ with $\left|\mathbf{n}\right|\leqslant\ell$, we have
\begin{align}\label{sieq:recurrence}
\rho_{\mathbf{n},\ell}^{N} &= \begin{cases}
\displaystyle 0 & \left|\mathbf{n}\right| = 0 ; \\
\displaystyle 1 & \left|\mathbf{n}\right| = \ell ; \\
\displaystyle \sum_{i=1}^{m} Q_{i}^{-}\left(\mathbf{n}\right) \rho_{\mathbf{n}_{i}^{-},\ell}^{N} + \sum_{i=1}^{m} Q_{i}^{+}\left(\mathbf{n}\right) \rho_{\mathbf{n}_{i}^{+},\ell}^{N} & 0<\left|\mathbf{n}\right| <\ell .
\end{cases}
\end{align}
Since $Q_{i}^{\pm}\left(\mathbf{n}\right)$ are continuous functions of $N$ with limits for $i=1,\dots ,m$, it follows from the fact that $\rho_{\mathbf{n},\ell}^{N}$ is a rational function of $\left\{Q_{i}^{\pm}\left(\mathbf{n}\right)\right\}_{i=1}^{m}$ \citep[see][Appendix~A]{mcavoy:JRSI:2015} that $\lim_{N\rightarrow\infty}\rho_{\mathbf{n},\ell}^{N}$ exists for any $\mathbf{n}$ with $0<\left|\mathbf{n}\right| <\ell$. Letting $N\rightarrow\infty$ in \textbf{Eq.~\ref{sieq:recurrence}} and using \textbf{Eq.~\ref{sieq:Qtransitions}} gives the following expression for $\rho_{\mathbf{n},\ell}^{\infty}\coloneqq\lim_{N\rightarrow\infty}\rho_{\mathbf{n},\ell}^{N}$:
\begin{align}\label{sieq:recurrenceInfinity}
\rho_{\mathbf{n},\ell}^{\infty} &= \begin{cases}
\displaystyle 0 & \left|\mathbf{n}\right| = 0 ; \\
\displaystyle 1 & \left|\mathbf{n}\right| = \ell ; \\
\displaystyle \sum_{i=1}^{m} \widetilde{Q}_{i}^{-}\left(\mathbf{n}\right) \rho_{\mathbf{n}_{i}^{-},\ell}^{\infty} + \sum_{i=1}^{m} \widetilde{Q}_{i}^{+}\left(\mathbf{n}\right) \rho_{\mathbf{n}_{i}^{+},\ell}^{\infty} & 0<\left|\mathbf{n}\right| <\ell .
\end{cases}
\end{align}
As a result, we see from our analysis of \textbf{Eq.~\ref{sieq:extinctionRecurrence}} that $\lim_{N\rightarrow\infty}\rho_{\mathbf{n},\ell}^{N}=1-\mathcal{E}_{\mathbf{n}}^{\ell}$, which completes the proof.
\end{proof}

We now sketch a proof of the following limit:
\begin{align}
\lim_{N\rightarrow\infty}\rho_{\mathbf{e}_{i},\mathbf{A}}^{N} &= \begin{cases}0 & \overline{a}\leqslant\overline{b} , \\ 1-\gamma_{i}^{\ast} & \overline{a}>\overline{b} .\end{cases}
\end{align}

Consider the first-visit distribution, $\mu_{\mathbf{e}_{i},\ell}$, on $\left\{\mathbf{n}\in S\ :\ \left|\mathbf{n}\right| =\ell\right\}$. Specifically, for $\mathbf{n}\in S$ with $\left|\mathbf{n}\right| =\ell$,
\begin{align}
\mu_{\mathbf{e}_{i},\ell}\left(\mathbf{n}\right) &\coloneqq \mathbb{P}\left[ Y_{\tau_{\ell}}=\mathbf{n}\ :\ Y_{0}=\mathbf{e}_{i} ,\ \tau_{\ell}<\infty \right] . \label{sieq:firstVisitDist}
\end{align}
Using this distribution, we can write a mutant's fixation probability as
\begin{align}\label{sieq:decomposeFP}
\rho_{\mathbf{e}_{i},\mathbf{A}}^{N} &= \rho_{\mathbf{e}_{i},\ell}^{N}\sum_{\substack{\mathbf{n}\in S \\ \left|\mathbf{n}\right| =\ell}}\mu_{\mathbf{e}_{i},\ell}\left(\mathbf{n}\right)\rho_{\mathbf{n},\mathbf{A}}^{N} .
\end{align}
In particular, $\rho_{\mathbf{e}_{i},\mathbf{A}}^{N}\leqslant\rho_{\mathbf{e}_{i},\ell}^{N}$, so $\lim_{N\rightarrow\infty}\rho_{\mathbf{e}_{i},\mathbf{A}}^{N}=0$ whenever $\overline{a}\leqslant\overline{b}$ because $\lim_{\ell\rightarrow\infty}\lim_{N\rightarrow\infty}\rho_{\mathbf{e}_{i},\ell}^{N}=0$ (Lemma \ref{lem:extinctionFixation}).

Suppose now that $\overline{a}>\overline{b}$. Let $s\coloneqq\overline{a}/\overline{b}-1$, which is positive because $\overline{a}>\overline{b}$. Moreover, since
\begin{align}
\lim_{N\rightarrow\infty} \frac{\sum_{j=1}^{m}N_{j}\left(N\right) a_{j}}{\sum_{j=1}^{m}N_{j}\left(N\right) b_{j}} &= \overline{a}/\overline{b} ,
\end{align}
there exists $N^{\ast}$ for which $\frac{\sum_{j=1}^{m}N_{j}\left(N\right) a_{j}}{\sum_{j=1}^{m}N_{j}\left(N\right) b_{i}}>1+s/2$ whenever $N\geqslant N^{\ast}$. In what follows, we let $r\coloneqq 1+s/3$ and $r'\coloneqq 1+s/2$ so that $1<r<r'<\overline{a}/\overline{b}$. We also assume that $N$ is finite but at least $N^{\ast}$.

In the chain $\left\{Y_{t}\right\}_{t\geqslant 0}$, the probability of losing a mutant in state $\mathbf{n}\in S$ is
\begin{align}
L\left(\mathbf{n}\right) &\coloneqq \sum_{j=1}^{m}Q_{j}^{-}\left(\mathbf{n}\right) = \frac{\left|\mathbf{n}\right|\sum_{j=1}^{m}b_{j}\left(N_{j}-n_{j}\right)}{\left|\mathbf{n}\right|\sum_{j=1}^{m}b_{j}\left(N_{j}-n_{j}\right) + \left(N-\left|\mathbf{n}\right|\right)\sum_{j=1}^{m}a_{j}n_{j}} ,
\end{align}
and the probability of gaining a mutant in this state is simply $1-L\left(\mathbf{n}\right)$.

For $\ell <N$, consider again the stopping time $\tau_{\ell}\coloneqq\min\left\{t\geqslant 0\ :\ \left| Y_{t}\right| =\ell\right\}$, and let $\left\{Y_{t}^{\tau_{\ell}}\right\}_{t\geqslant 0}$ denote the stopped chain (i.e. $Y_{t}^{\tau_{\ell}}=Y_{\min\left\{t,\tau_{\ell}\right\}}$). In what follows, the notation $\mathbb{P}_{\mu}$ and $\mathbb{E}_{\mu}$ refers to the probability and expectation, respectively, when the chain $\left\{Y_{t}^{\tau_{\ell}}\right\}_{t\geqslant 0}$ has initial distribution $Y_{0}^{\tau_{\ell}}\sim\mu$. (If the subscript is a state, $\mathbf{n}$, instead of a distribution, then this notation indicates that the initial state of this chain is $\mathbf{n}$.) The main ingredient we will need to prove the lemma is to establish the existence of $\ell$ such that, for all $\ell '\geqslant\ell$,
\begin{align}
r^{-\ell '} &\geqslant \mathbb{E}_{\mu_{\mathbf{e}_{i},\ell '}}\left[ r^{-\left| Y_{t}^{\tau_{\ell}}\right|} \right] . \label{sieq:expInequality}
\end{align}

To establish \textbf{Eq.~\ref{sieq:expInequality}}, we first note that for $\left|\mathbf{n}\right| <N$,
\begin{align}
\frac{\sum_{j=1}^{m}a_{j}\frac{n_{j}}{\left|\mathbf{n}\right|}}{\sum_{j=1}^{m}b_{j}\frac{N_{j}-n_{j}}{N-\left|\mathbf{n}\right|}} > r' \iff \sum_{j=1}^{m} \left[ \left(1-\frac{\left|\mathbf{n}\right|}{N}\right) a_{j}+\frac{\left|\mathbf{n}\right|}{N}r'b_{j} \right] \frac{n_{j}}{\left|\mathbf{n}\right|} - r'\sum_{j=1}^{m}b_{j}\frac{N_{j}}{N} > 0 .
\end{align}
If $\left|\frac{n_{j}}{\left|\mathbf{n}\right|}-\frac{N_{j}}{N}\right| <\delta$ for every $j=1,\dots ,m$, then $\frac{n_{j}}{\left|\mathbf{n}\right|}>\frac{N_{j}}{N}-\delta$ in particular, which gives
\begin{align}
\sum_{j=1}^{m} &\left[ \left(1-\frac{\left|\mathbf{n}\right|}{N}\right) a_{j}+\frac{\left|\mathbf{n}\right|}{N}r'b_{j} \right] \frac{n_{j}}{\left|\mathbf{n}\right|} - r'\sum_{j=1}^{m}b_{j}\frac{N_{j}}{N} \nonumber \\
&> \sum_{j=1}^{m}\left[ \left(1-\frac{\left|\mathbf{n}\right|}{N}\right) a_{j}+\frac{\left|\mathbf{n}\right|}{N}r'b_{j} \right] \left(\frac{N_{j}}{N}-\delta\right) - r'\sum_{j=1}^{m}b_{j}\frac{N_{j}}{N} \nonumber \\
&= \left(1-\frac{\left|\mathbf{n}\right|}{N}\right)\left(\sum_{j=1}^{m}a_{j}\frac{N_{j}}{N}-r'\sum_{j=1}^{m}b_{j}\frac{N_{j}}{N}\right) - \delta\sum_{j=1}^{m}\left[ \left(1-\frac{\left|\mathbf{n}\right|}{N}\right) a_{j}+\frac{\left|\mathbf{n}\right|}{N}r'b_{j} \right] .
\end{align}
It follows that if $\delta$ is a fixed real number satisfying
\begin{align}
0 < \delta < \frac{\sum_{j=1}^{m}a_{j}\frac{N_{j}}{N}-r'\sum_{j=1}^{m}b_{j}\frac{N_{j}}{N}}{\sum_{j=1}^{m}\left[ a_{j}+\left(N-1\right) r'b_{j} \right]} , \label{sieq:deltaBound}
\end{align}
then $\frac{\sum_{j=1}^{m}a_{j}\frac{n_{j}}{\left|\mathbf{n}\right|}}{\sum_{j=1}^{m}b_{j}\frac{N_{j}-n_{j}}{N-\left|\mathbf{n}\right|}} > r'$ whenever $\left|\frac{n_{j}}{\left|\mathbf{n}\right|}-\frac{N_{j}}{N}\right| <\delta$ for every $j=1,\dots ,m$. Note that there exists such a $\delta$ in the range required by \textbf{Eq.~\ref{sieq:deltaBound}} because of our assumption that $N\geqslant N^{\ast}$, i.e. $\frac{\sum_{j=1}^{m}a_{j}N_{j}\left(N\right)}{\sum_{j=1}^{m}b_{j}N_{j}\left(N\right)}>r'$.

In every non-absorbing state, the probability of a mutant-type birth is bounded from below by some $p_{\ast}>0$ and above by some $p^{\ast}<1$, so it is possible for the chain to transition between any two non-absorbing states in finitely many steps. For every mutant (resp. resident) birth, the number of mutants in environment $j$ is increased (resp. decreased) by one with probability $\frac{N_{j}-n_{j}}{N-\left|\mathbf{n}\right|}$ (resp. $\frac{n_{j}}{\left|\mathbf{n}\right|}$); see \textbf{Eq.~\ref{sieq:scaledTransitions}}. Moreover, $\frac{n_{j}}{\left|\mathbf{n}\right|}\leqslant\frac{N_{j}-n_{j}}{N-\left|\mathbf{n}\right|}$ if and only if $\frac{n_{j}}{\left|\mathbf{n}\right|}\leqslant\frac{N_{j}}{N}$, which means that a mutant offspring is at least (resp. at most) as likely to replace a resident as a resident offspring is to replace a mutant in environment $j$ when $\frac{n_{j}}{\left|\mathbf{n}\right|}\leqslant\frac{N_{j}}{N}$ (resp. $\frac{n_{j}}{\left|\mathbf{n}\right|}\geqslant\frac{N_{j}}{N}$). A balance between the two is achieved when $\frac{n_{j}}{\left|\mathbf{n}\right|}=\frac{N_{j}}{N}$. Furthermore, if $k\neq j$, then a new mutant in environment $k$ increases the fraction $\frac{N_{j}-n_{j}}{N-\left|\mathbf{n}\right|}$, while a new resident in environment $k$ increases the fraction $\frac{n_{j}}{\left|\mathbf{n}\right|}$.

Let $\left(Y_{t}^{\tau_{\ell}}\right)_{j}$ denote the number of mutant-type individuals in environment $j$ (i.e. $n_{j}$ when $Y_{t}^{\tau_{\ell}}=\mathbf{n}$). Fix $\delta ,\epsilon >0$. From the heuristic in the previous paragraph, one can show that if $0<\xi <\left(1-r^{-2}\right)\left(\frac{1}{1+r}-\frac{1}{1+r'}\right)$, then there exists $\ell\geqslant 1$ such that whenever \textit{(i)} $\ell '\geqslant\ell$, \textit{(ii)} $\ell\leqslant k<N$, and \textit{(iii)} $n\geqslant 0$, we have
\begin{align}
\mathbb{P}_{\mu_{\mathbf{e}_{i},\ell '}} &\left[ r^{-1} - L\left(Y_{t}^{\tau_{\ell}}\right) - \left(1-L\left(Y_{t}^{\tau_{\ell}}\right)\right) r^{-2} > \xi \ :\ \left| Y_{t}^{\tau_{\ell}}\right| =k \right] \nonumber \\
&= \mathbb{P}_{\mu_{\mathbf{e}_{i},\ell '}}\left[ L\left(Y_{t}^{\tau_{\ell}}\right) < \frac{1}{1+r} - \frac{\xi}{1-r^{-2}} \ :\ \left| Y_{t}^{\tau_{\ell}}\right| =k \right] \nonumber \\
&\geqslant \mathbb{P}_{\mu_{\mathbf{e}_{i},\ell '}}\left[ \frac{\sum_{j=1}^{m}a_{j}\frac{\left(Y_{t}^{\tau_{\ell}}\right)_{j}}{\left| Y_{t}^{\tau_{\ell}}\right|}}{\sum_{j=1}^{m}b_{j}\frac{N_{j}-\left(Y_{t}^{\tau_{\ell}}\right)_{j}}{N-\left| Y_{t}^{\tau_{\ell}}\right|}} > r' \ :\ \left| Y_{t}^{\tau_{\ell}}\right| =k \right] \nonumber \\
&\geqslant \mathbb{P}_{\mu_{\mathbf{e}_{i},\ell '}}\left[ \sum_{j=1}^{m}\left| \frac{\left(Y_{t}^{\tau_{\ell}}\right)_{j}}{\left|Y_{t}^{\tau_{\ell}}\right|}-\frac{N_{j}}{N}\right| <\delta \ :\ \left| Y_{t}^{\tau_{\ell}}\right| =k \right] \nonumber \\
&> 1-\varepsilon .
\end{align}
(Again, the subscript in $\mathbb{P}_{\mu_{\mathbf{e}_{i},\ell '}}$ indicates that $Y_{0}^{\tau_{\ell}}\sim\mu_{\mathbf{e}_{i},\ell '}$.) Letting $0<\varepsilon <\frac{\xi}{\xi + 1-r^{-1}}$, we find that
\begin{align}
\mathbb{E}_{\mu_{\mathbf{e}_{i},\ell '}} &\left[ r^{-1} - L\left(Y_{t}^{\tau_{\ell}}\right) - \left(1-L\left(Y_{t}^{\tau_{\ell}}\right)\right) r^{-2} \ :\ \left| Y_{t}^{\tau_{\ell}}\right| =k \right] \nonumber \\
&\geqslant \xi\left(1-\varepsilon\right) + \left(r^{-1}-1\right)\varepsilon \nonumber \\
&\geqslant 0 . \label{sieq:eachCountPositive}
\end{align}
Note that in the arguments preceding \textbf{Eq.~\ref{sieq:eachCountPositive}}, we assumed that $k<N$. However, if $\left| Y_{t}^{\tau_{\ell}}\right| =N$, then $L\left(Y_{t}^{\tau_{\ell}}\right) =0$ and we have $\mathbb{E}_{\mu_{\mathbf{e}_{i},\ell '}}\left[ r^{-1} - L\left(Y_{t}^{\tau_{\ell}}\right) - \left(1-L\left(Y_{t}^{\tau_{\ell}}\right)\right) r^{-2} \ :\ \left| Y_{t}^{\tau_{\ell}}\right| =N \right] =r^{-1}-r^{-2}>0$.

From \textbf{Eq.~\ref{sieq:eachCountPositive}}, it follows that
\begin{align}
\mathbb{E}_{\mu_{\mathbf{e}_{i},\ell '}} &\left[r^{-\left| Y_{t}^{\tau_{\ell}}\right|} - \mathbb{E}\left[ r^{-\left|Y_{t+1}^{\tau_{\ell}}\right|} \ :\ Y_{t}^{\tau_{\ell}} \right] \right] \nonumber \\
&= \sum_{k=\ell}^{N} \mathbb{E}_{\mu_{\mathbf{e}_{i},\ell '}}\left[r^{-\left| Y_{t}^{\tau_{\ell}}\right|} - \mathbb{E}\left[ r^{-\left|Y_{t+1}^{\tau_{\ell}}\right|} \ :\ Y_{t}^{\tau_{\ell}} \right] \ :\ \left| Y_{t}^{\tau_{\ell}}\right| =k \right] \mathbb{P}_{\mu_{\mathbf{e}_{i},\ell '}}\left[ \left| Y_{t}^{\tau_{\ell}}\right| =k \right] \nonumber \\
&= \sum_{k=\ell}^{N} r^{-\left(k-1\right)} \mathbb{E}_{\mu_{\mathbf{e}_{i},\ell '}}\left[r^{-1} - L\left(Y_{t}^{\tau_{\ell}}\right) - \left(1-L\left(Y_{t}^{\tau_{\ell}}\right)\right) r^{-2} \ :\ \left| Y_{t}^{\tau_{\ell}}\right| =k \right] \mathbb{P}_{\mu_{\mathbf{e}_{i},\ell '}}\left[ \left| Y_{t}^{\tau_{\ell}}\right| =k \right] \nonumber \\
&\geqslant 0 .
\end{align}
By induction, we then obtain the desired inequality, $r^{-\ell '}\geqslant\mathbb{E}_{\mu_{\mathbf{e}_{i},\ell '}}\left[ r^{-\left| Y_{t}^{\tau_{\ell}}\right|} \right]$ (\textbf{Eq.~\ref{sieq:expInequality}}). Furthermore, since the Markov chain $\left\{Y_{t}^{\tau_{\ell}}\right\}_{t\geqslant 0}$ is finite, we can take the limit of \textbf{Eq.~\ref{sieq:expInequality}} as $t\rightarrow\infty$ to get
\begin{align}
r^{-\ell '} &\geqslant \sum_{\substack{\mathbf{n}\in S \\ \left|\mathbf{n}\right| =\ell '}} \mu_{\mathbf{e}_{i},\ell '}\left(\mathbf{n}\right) \Big(\mathbb{P}_{\mathbf{n}}\left[\left| Y_{\tau_{\ell}}\right| =N\right] r^{-N}+\left(1-\mathbb{P}_{\mathbf{n}}\left[\left| Y_{\tau_{\ell}}\right| =N\right]\right) r^{-\ell}\Big) . \label{sieq:limitingInequality}
\end{align}

\textbf{Eq.~\ref{sieq:limitingInequality}} holds for all $\ell '\geqslant\ell$, which means, in particular, we can let $\ell '=2\ell$ to see that
\begin{align}
1-\sum_{\substack{\mathbf{n}\in S \\ \left|\mathbf{n}\right| =2\ell}} \mu_{\mathbf{e}_{i},2\ell}\left(\mathbf{n}\right)\mathbb{P}_{\mathbf{n}}\left[\left| Y_{\tau_{\ell}}\right| =N\right] &\leqslant r^{-\ell} - r^{-\left(N-\ell\right)}\sum_{\substack{\mathbf{n}\in S \\ \left|\mathbf{n}\right| =2\ell}} \mu_{\mathbf{e}_{i},2\ell}\left(\mathbf{n}\right) \mathbb{P}_{\mathbf{n}}\left[\left| Y_{\tau_{\ell}}\right| =N\right]
\end{align}
whenever $\ell$ is sufficiently large. Thus, $\lim_{\ell\rightarrow\infty}\lim_{N\rightarrow\infty}\sum_{\substack{\mathbf{n}\in S \\ \left|\mathbf{n}\right| =2\ell}} \mu_{\mathbf{e}_{i},2\ell}\left(\mathbf{n}\right)\mathbb{P}_{\mathbf{n}}\left[\left| Y_{\tau_{\ell}}\right| =N\right] =1$. Since
\begin{align}
\mathbb{P}_{\mathbf{n}}\left[\left| Y_{\tau_{\ell}}\right| =N\right] &\leqslant \rho_{\mathbf{n},\mathbf{A}}^{N} ,
\end{align}
we also have $\lim_{\ell\rightarrow\infty}\lim_{N\rightarrow\infty}\sum_{\substack{\mathbf{n}\in S \\ \left|\mathbf{n}\right| =2\ell}} \mu_{\mathbf{e}_{i},2\ell}\left(\mathbf{n}\right)\rho_{\mathbf{n},\mathbf{A}}^{N}=1$. Therefore, by \textbf{Eq.~\ref{sieq:decomposeFP}} and Lemma \ref{lem:extinctionFixation},
\begin{align}
\lim_{N\rightarrow\infty}\rho_{\mathbf{e}_{i},\mathbf{A}}^{N} &= \lim_{\ell\rightarrow\infty}\lim_{N\rightarrow\infty}\rho_{\mathbf{e}_{i},2\ell}^{N}\sum_{\substack{\mathbf{n}\in S \\ \left|\mathbf{n}\right| =2\ell}}\mu_{\mathbf{e}_{i},2\ell}\left(\mathbf{n}\right)\rho_{\mathbf{n},\mathbf{A}}^{N} \nonumber \\
&= \lim_{\ell\rightarrow\infty}\rho_{\mathbf{e}_{i},2\ell}^{\infty} \nonumber \\
&= 1-\gamma_{i}^{\ast} .
\end{align}

%%%%

\section{Linear dispersal structures}\label{sec:appendixB}
In the main text, we assumed that the dispersal structure was represented by a complete graph. By ignoring dispersal heterogeneity, we could focus on the effects of environmental fitness heterogeneity on a mutant's fixation probability. On a complete graph, moments of the fitness distributions for the mutant and resident types, including the arithmetic mean and standard deviation, determine the fate of a rare mutant.

When environmental heterogeneity is generalized to arbitrary dispersal graphs, where individuals see potentially only a small number of neighbors, the evolutionary dynamics become more complex. In this case, the distribution of fitness values is still important, but it is also matters where different environments are located relative to each other. Thus, the general question of how environmental fitness heterogeneity affects the fate of a mutant is determined by both the moments of the individual fitness distributions and the spatial correlations between these values. A thorough analysis of these models is outside the scope of this paper.

However, to illustrate the difference in the dynamics and to compare with the results on the complete graph, we consider a bimodal distribution of fitness values of a cycle. A cycle is a one-dimensional, periodic spatial structure in which every individual has exactly two neighbors \citep{ohtsuki:PRSB:2006}. As before, $a_{i} \in \left\{a_{1},a_{2}\right\}$ and $b_{i} \in \left\{b_{1},b_{2}\right\}$ for some $a_{1},a_{2},b_{1},b_{2}>0$. We assume a uniform, spatially-periodic distribution of fitness values, such that for every node in environment $1$, the two neighboring nodes are in environment $2$, and vice versa (see \textbf{Fig.~\ref{fig:cycle}}). (This distribution is in fact a good estimate for the evolution on a cycle with random fitness values derived from the same bimodal distribution.)

\begin{figure}
\centering
\includegraphics[width=0.35\textwidth]{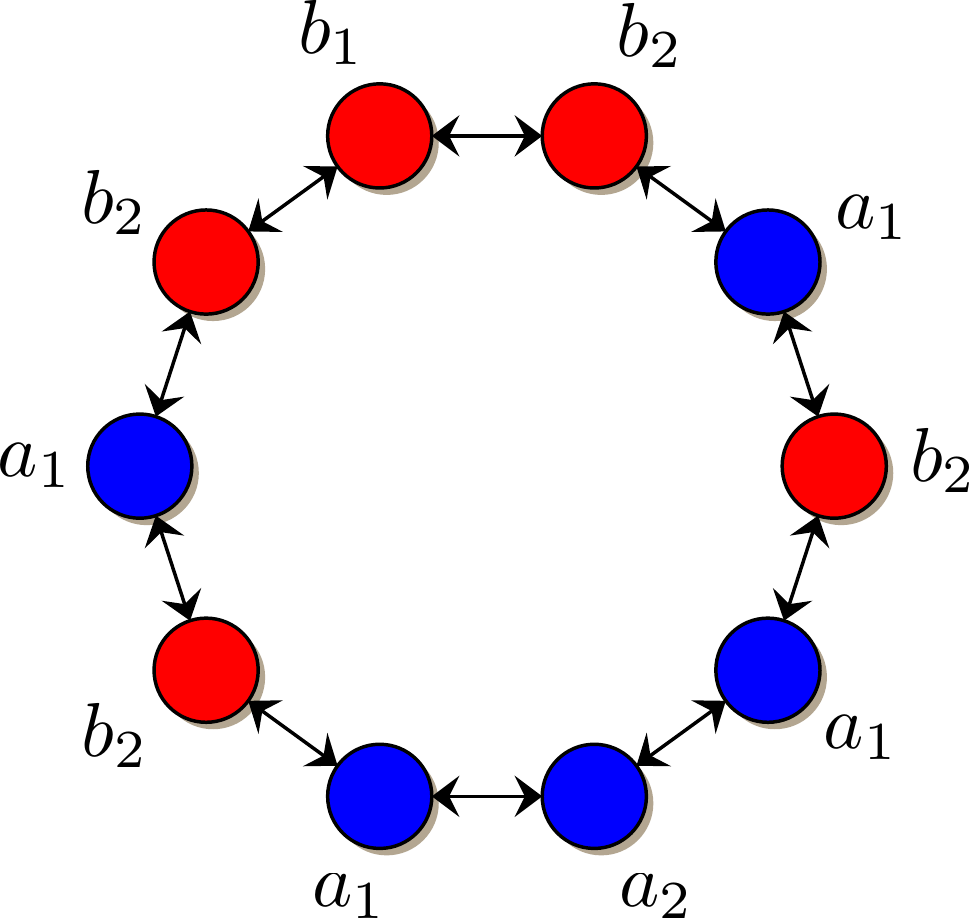}
\caption{Environmental heterogeneity on the cycle. The dispersal graph, a cycle, is a one-dimensional, periodic structure, meaning an individual's offspring can be propagated only to one of its two neighbors. In the case we consider, the environments are alternating, so that every environment of type $1$ has two neighbors of type $2$, and vice versa.\label{fig:cycle}}
\end{figure}

We have solved the Kolmogorov backward equation numerically for this model. If mutant fitness varies while resident fitness is constant over all spatial locations (i.e. $a_{1,2}= \bar{a}\pm\Delta_{a}$ and $b_{1,2} = 1$), then we observe that a mutant's fixation probability is decreased as the mean fitness of the mutants is kept constant and the standard deviation of the fitness values is increased (just like when the dispersal structure is a complete graph). The effect on fixation probability, however, is more significant than it is in the case of a complete dispersal graph. In fact, as $\Delta_{a} \to \Delta_{\rm max} =\bar{a}$, the fixation probability, $\rho_{A}({\bf a},{\bf b})$, approaches zero (see \fig{cycleSuppression}). The fitness parameters are chosen similar to the complete graph, and the population size is $N=100$.

\begin{figure}
\centering
\includegraphics[width=0.8\textwidth]{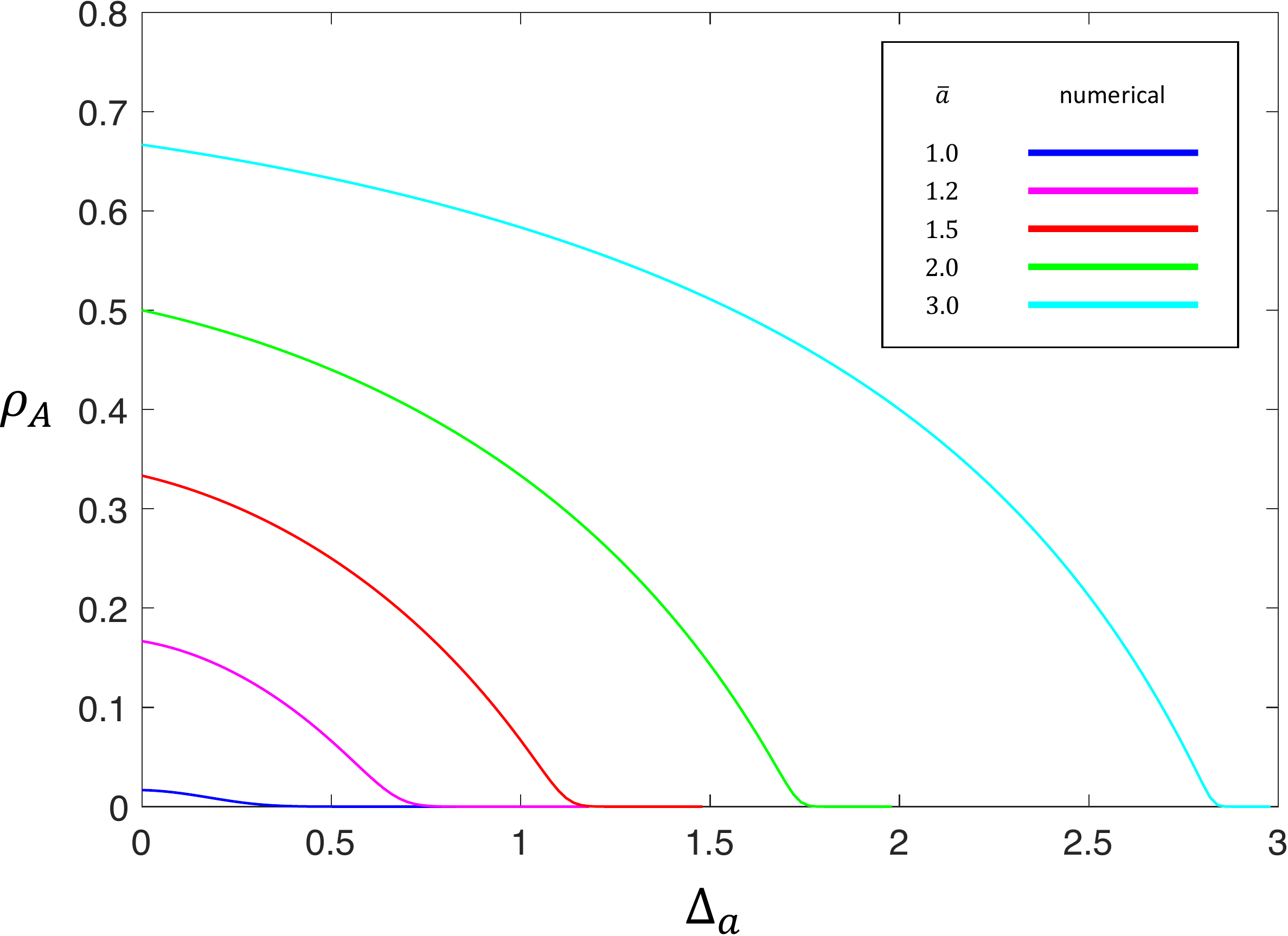}
\caption{Fixation probability of the mutant type, $A$, as a function of (half) the width of the mutant's fitness distribution, $\Delta_{a}$. The fitness values for the mutant are either $\overline{a}-\Delta_{a}$ or $\overline{a}+\Delta_{a}$. The fitness values of the resident are set to unity and do not change across population. Every location in environment $1$ (resp. $2$) is a neighbor to two individuals in environment $2$ (resp. $1$). The population size is $N=100$, and $\overline{a}\in\left\{1.0,1.2,1.5,2.0,3.0\right\}$. The results are obtained from exact solutions of the Kolmogorov equation for fixation probability. As $\Delta_{a}$ grows, a mutant's fixation probability decreases, consistent with suppression. However, the effect is more significant than it is when the dispersal structure is a complete graph.\label{fig:cycleSuppression}}
\end{figure}

We also considered the effects of heterogeneity in resident fitness. We let $a_{1,2}=\bar{a}$ and $b_{1,2}=1 \pm \Delta_{b}$. Just as we observed for the complete graph, resident heterogeneity now increases the fixation probability of a randomly-placed mutant. However, now this effect is not restricted to small population sizes. In \fig{cycleAmplification}, the results are shown for $N=100$ and various values of mutant fitness ($\bar{a} =0.8,0.9,1.0,1.1$). Curiously, resident heterogeneity increases the fixation probability of deleterious mutants ($\bar{a}=0.8,0.9$, for example). For large enough values of $\Delta_{b}$, deleterious mutants (in a uniform environment) become strongly advantageous, which can be seen in \fig{cycleAmplification}.

\begin{figure}
\centering
\includegraphics[width=0.8\textwidth]{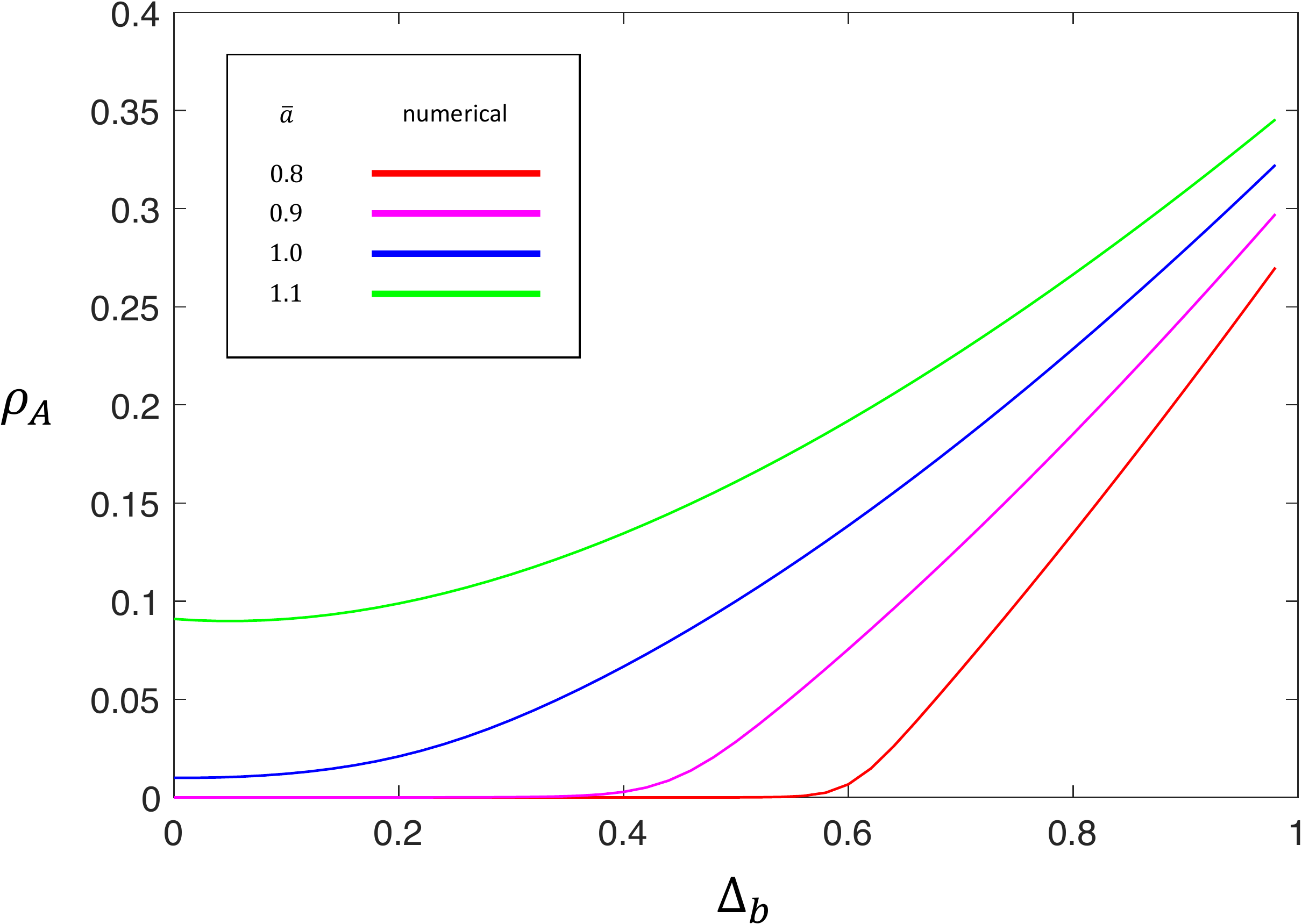}
\caption{Fixation probability of the mutant type, $A$, as a function of (half) the width of the resident's fitness distribution, $\Delta_{b}$. The fitness values for the resident are either $\overline{b}-\Delta_{b}$ or $\overline{b}+\Delta_{b}$, where $\bar{b}=1$. Again, every location in environment $1$ (resp. $2$) is a neighbor to two individuals in environment $2$ (resp. $1$). The population size is $N=100$, and $\overline{a}\in\left\{0.8,0.9,1.0,1.1\right\}$. The results are obtained from exact solutions of the Kolmogorov equation for fixation probability. As $\Delta_{b}$ grows, a near-neutral mutant's fixation probability increases, consistent with amplification.\label{fig:cycleAmplification}}
\end{figure}

\section*{Acknowledgments}
We are grateful to the referees for many constructive comments on earlier versions of the manuscript. This project was supported by the Office of Naval Research, grant N00014-16-1-2914.

\end{document}